\documentclass{llncs}

\usepackage{etex}
\usepackage{amssymb,amsmath,latexsym}
\usepackage{graphicx,color}
\usepackage{wrapfig}
\usepackage{subfig}
\usepackage{times}
\usepackage{paralist}

\usepackage[long]{optional}

\newcommand{\gcom}[1]{}
\newcommand{\gcomb}[1]{}

\usepackage{etoolbox}
\usepackage{tikz}
\usetikzlibrary{positioning}
\usetikzlibrary{patterns}
\usetikzlibrary{arrows}
\usetikzlibrary{scopes}
\usetikzlibrary{backgrounds,fit}
\usetikzlibrary{calc}
\usetikzlibrary{decorations.markings}
\usetikzlibrary{shapes}

\newcommand{\tikzArrowToText}[1]{%
\mathop{\begin{tikzpicture}[baseline={([yshift=-0.5ex]current 
bounding box.south)}]%
\draw[#1] (0,0) -- (1.1em,0);%
\end{tikzpicture}}}

\tikzset{nodelabel/.style={font=\footnotesize}}
\tikzset{edgelabel/.style={font=\footnotesize}}
\tikzset{henodenumber/.style={font=\footnotesize}}
\tikzset{undiredge/.style={shorten >=0pt, shorten <=0pt}}
\tikzset{partialedge/.style={-left to, thick}}
\tikzset{dotline/.style={dotted, shorten <=5pt, shorten >=5pt}}

\tikzset{graphedge/.style={->, >=latex}}
\tikzset{bigraphedge/.style={<->, >=latex}}
\tikzset{invgraphedge/.style={<-, >=latex}}
\tikzset{henode/.style={draw, rectangle, rounded corners=2.5pt}}
\tikzset{stdnode/.style={draw, circle}}
\tikzset{mapedge/.style={->, >=stealth', dashed, color=orange}}

\tikzset{uqelement/.style={dashed}}
\tikzset{uqnode/.style={fill=gray}}

\tikzset{ruleappedge/.style={double, ->}}
\tikzset{mor-tot-inj/.style={catt-catt}}
\tikzset{mor-tot/.style={-catt}}
\tikzset{mor-parr/.style={-catpr}}
\tikzset{mor-parl/.style={-catpl}}
\tikzset{mor-minor/.style={|-catt}}
\tikzset{mor-subgraph/.style={)-catt}}
\tikzset{mor-indsubgraph/.style={open triangle 45 reversed-catt}}
\tikzset{mor-genorder/.style={]-catt}}
\tikzset{noedge/.style={decoration={markings,mark=at position 0.5 with 
{\arrow[scale=1.6,color=red]{|}}}, postaction={decorate}}}
\tikzset{closure/.style={postaction={decorate, decoration={raise=4pt, 
markings, mark=at position 1 with {\node[scale=0.65]{\textbf{*}};}}}}}
\tikzset{morlabel/.style={midway, inner sep=2pt}}

\tikzset{backstepedge/.style={-open triangle 45, shorten <=0pt, shorten >=0pt}}

\tikzset{graphbox/.style={draw, dashed, rounded corners=2mm, outer sep=5pt, on 
background layer}}
\tikzset{graphboxgrey/.style={draw, dashed, black!50, rounded corners=2mm, 
outer sep=5pt, on background layer}}

\newdimen\arrowsize
\pgfarrowsdeclare{recatt}{catt}
{
\arrowsize=1pt
\advance\arrowsize by .5\pgflinewidth
\pgfarrowsleftextend{-2.5\arrowsize-.5\pgflinewidth}
\pgfarrowsrightextend{.5\pgflinewidth}
}
{
\pgfsetlinewidth{0.4pt}
\arrowsize=0.8pt
\advance\arrowsize by .5\pgflinewidth
\pgfsetdash{}{0pt} 
\pgfsetroundjoin 
\pgfsetroundcap 
\pgfpathmoveto{\pgfpointorigin}
\pgfpatharc{270}{190}{2.8\arrowsize}
\pgfusepathqstroke
\pgfpathmoveto{\pgfpointorigin}
\pgfpatharc{90}{170}{2.8\arrowsize}
\pgfusepathqstroke
}

\pgfarrowsdeclarereversed{catt}{recatt}{recatt}{catt}
\pgfarrowsdeclaredouble{recatcatt}{catcatt}{recatt}{catt}
\pgfarrowsdeclaredouble{catcatt}{recatcatt}{catt}{recatt}

\pgfarrowsdeclare{helperpr}{helperpr}
{
\arrowsize=1pt
\advance\arrowsize by .5\pgflinewidth
\pgfarrowsleftextend{-2.5\arrowsize-.5\pgflinewidth}
\pgfarrowsrightextend{.5\pgflinewidth}
}
{
\pgfsetlinewidth{0.55pt}
\arrowsize=0.8pt
\advance\arrowsize by .5\pgflinewidth
\pgfsetdash{}{0pt} 
\pgfsetroundjoin 
\pgfsetroundcap 
\pgfpathmoveto{\pgfpointorigin}
\pgfpatharc{90}{170}{3\arrowsize}
\pgfusepathqstroke
}

\pgfarrowsdeclare{helperpl}{helperpl}
{
\arrowsize=1pt
\advance\arrowsize by .5\pgflinewidth
\pgfarrowsleftextend{-2.5\arrowsize-.5\pgflinewidth}
\pgfarrowsrightextend{.5\pgflinewidth}
}
{
\pgfsetlinewidth{0.55pt}
\arrowsize=0.8pt
\advance\arrowsize by .5\pgflinewidth
\pgfsetdash{}{0pt} 
\pgfsetroundjoin 
\pgfsetroundcap 
\pgfpathmoveto{\pgfpointorigin}
\pgfpatharc{270}{190}{3\arrowsize}
\pgfusepathqstroke
}

\pgfarrowsdeclarereversed{rehelperpr}{rehelperpr}{helperpr}{helperpr}
\pgfarrowsdeclarereversed{rehelperpl}{rehelperpl}{helperpl}{helperpl}
\pgfarrowsdeclarealias{catpl}{catpl}{helperpl}{helperpl}
\pgfarrowsdeclarealias{catpr}{catpr}{helperpr}{helperpr}

\makeatletter
\@ifclassloaded{beamer}{%
\newcommand<>{\StaticGraphbox}[5][grbox]{%
\def\hAlignCJS{right}%
\def\vAlignCJS{top}%
\def\nameCJS{}%
\renewcommand{\do}[1]{%
  \ifthenelse{\equal{##1}{top}}{\def\vAlignCJS{top}}{%
  \ifthenelse{\equal{##1}{bottom}}{\def\vAlignCJS{bottom}}{%
  \ifthenelse{\equal{##1}{right}}{\def\hAlignCJS{right}}{%
  \ifthenelse{\equal{##1}{left}}{\def\hAlignCJS{left}}{\def\nameCJS{##1}}}}}}%
  \docsvlist{#5}%
  \begin{pgfonlayer}{background}
    \node[circle, above left= #4 and #3 of #2] (top_left) {};
    \node[circle, above right= #4 and #3 of #2] (top_right) {};
    \node[circle, below left= #4 and #3 of #2] (bottom_left) {};
    \node[circle, below right= #4 and #3 of #2] (bottom_right) {};
    \node[fit=(top_left) (bottom_right), outer sep=5pt] (#1) {} ;
    \draw#6[dashed, black!50,rounded corners=2mm] 
    ($(#1.north west) + (5pt,-5pt)$) -- ($(#1.north east) + (-5pt,-5pt)$) --
    ($(#1.south east) + (-5pt,5pt)$) -- ($(#1.south west) + (5pt,5pt)$) -- 
    cycle ;
    \node#6[left] (boxname) at (\vAlignCJS_\hAlignCJS.east) {\nameCJS};
  \end{pgfonlayer}
}}{%
\newcommand{\StaticGraphbox}[5][grbox]{%
\def\hAlignCJS{right}%
\def\vAlignCJS{top}%
\def\nameCJS{}%
\renewcommand{\do}[1]{%
  \ifthenelse{\equal{##1}{top}}{\def\vAlignCJS{top}}{%
  \ifthenelse{\equal{##1}{bottom}}{\def\vAlignCJS{bottom}}{%
  \ifthenelse{\equal{##1}{right}}{\def\hAlignCJS{right}}{%
  \ifthenelse{\equal{##1}{left}}{\def\hAlignCJS{left}}{\def\nameCJS{##1}}}}}}%
  \docsvlist{#5}%
  \begin{pgfonlayer}{background}
    \node[circle, above left= #4 and #3 of #2] (top_left) {};
    \node[circle, above right= #4 and #3 of #2] (top_right) {};
    \node[circle, below left= #4 and #3 of #2] (bottom_left) {};
    \node[circle, below right= #4 and #3 of #2] (bottom_right) {};
    \node[fit=(top_left) (bottom_right), outer sep=5pt] (#1) {} ;
    \draw[dashed, black!50,rounded corners=2mm] 
    ($(#1.north west) + (5pt,-5pt)$) -- ($(#1.north east) + (-5pt,-5pt)$) --
    ($(#1.south east) + (-5pt,5pt)$) -- ($(#1.south west) + (5pt,5pt)$) -- 
    cycle ;
    \node (boxname) at (\vAlignCJS_\hAlignCJS) {\nameCJS};
  \end{pgfonlayer}
}}
\makeatother

\usepackage{etoolbox}
\usepackage{hyperref}
\usepackage{ifthen}
\usepackage{rotating}
\usepackage{marginnote}

\newcommand{\nat}{\ensuremath{\mathbb{N}}}

\renewcommand{\phi}{\varphi}
\renewcommand{\epsilon}{\varepsilon}
\newcommand{\pto}{\rightharpoonup}

\newcommand{\ito}{\tikzArrowToText{mor-tot-inj}}

\newcommand{\upclosed}[1]{\mathord{\uparrow} #1}

\newcommand{\rotleq}{\begin{sideways}$\leq$\end{sideways}}

\newenvironment{proposition_app}[2][]{\noindent{\bf \hyperref[#2]{Proposition 
\ref*{#2}}\ifthenelse{\equal{#1}{}}{.}{ (#1).}}\it}{}
\newenvironment{lemma_app}[2][]{\noindent{\bf \hyperref[#2]{Lemma 
\ref*{#2}}\ifthenelse{\equal{#1}{}}{.}{ (#1).}}\it}{}

\newbool{bshowextranotes}

\newcommand{\extranote}[1]{\ifbool{bshowextranotes}{\marginnote{\textnormal{#1}}}{}}

\spnewtheorem{procedure}{Procedure}{\bfseries}{\rmfamily}
\newcommand{\arity}{\mathit{ar}}
\newcommand{\ignore}[1]{}
\newcommand{\subArrow}{\tikzArrowToText{mor-subgraph}}
\newcommand{\subOrder}{\subseteq}
\newcommand{\instAdd}{\diamond}
\newcommand{\qNodes}{\mathit{qn}}
\newcommand{\pseudoParagraph}[1]{\noindent\textit{#1:} }
\newcommand{\boundedPathk}{\mathcal{G}_k}
\newcommand{\instBound}[2][]{\ensuremath{\mathit{bound}\ifthenelse{\equal{#1}{}}{}{_{#1}}(#2)}}
\newcommand{\predBasis}[2][]{\ensuremath{\mathit{pb}\ifthenelse{\equal{#1}{}}{}{_{#1}}(#2)}}
\newcommand{\pred}[2][]{\ensuremath{\mathit{Pred}\ifthenelse{\equal{#1}{}}{}{_{#1}}(#2)}}
\newcommand{\predAll}[1]{\ensuremath{\mathit{Pred}^*(#1)}}

\newcommand*{\ExampleRuleScale}{0.9}
\tikzset{StdGraphGrid/.style={x=1.6cm, y=-1.5cm}}

\title{Parameterized Verification of Graph Transformation Systems with Whole 
Neighbourhood Operations\thanks{Research partially supported by DFG project 
GaReV.\opt{long}{ This paper is an extended version of \cite{DS2014} 
additionally containing the proofs.}}}

\author{Giorgio Delzanno\inst{1} \and Jan St\"uckrath\inst{2}}

\institute{
  Univerit\`a di Genova, Italy\\
  \email{giorgio.delzanno@unige.it}
  \smallskip \and
  Universit\"at~Duisburg-Essen, Germany\\
  \email{jan.stueckrath@uni-due.de}}
  
\begin{document}

\maketitle

\begin{abstract}
We introduce a new class of graph transformation systems in which rewrite 
rules can be guarded by universally quantified conditions on the neighbourhood 
of nodes. These conditions are defined via special graph patterns which may be 
transformed by the rule as well. For the new class for graph rewrite rules, we 
provide a symbolic procedure working on minimal representations of upward 
closed sets of configurations. We prove correctness and effectiveness of the 
procedure by a categorical presentation of rewrite rules as well as the 
involved order, and using results for well-structured transition systems. We 
apply the resulting procedure to the analysis of the Distributed Dining 
Philosophers protocol on an arbitrary network structure.
\end{abstract}

\section{Introduction}\label{sec:introduction}
Parameterized verification of distributed algorithms is a very challenging 
task. Distributed algorithms are often sensible to the network topology and 
they are based on communication patterns like broadcast messages and 
conditions on channels that can easily generate undecidable verification 
instances or finite-state problems of high combinatorial complexity. In order 
to naturally model interaction rules of topology-sensitive protocols it seems 
natural to consider languages based on graph rewriting and transformations as 
proposed in \cite{handbook-spo}. However, in this formalism rules can only 
match fixed subgraph in the graph they are applied to. Since we need to 
specify rules where the entire neighbourhood of a node is matched by the rule, 
we extend the standard approach by universally quantified patterns attached to 
nodes. With these patterns the matching of a left side of a rule can be 
increased until the entire neighbourhood of a node is covered. If the matching 
cannot be extended in this way the rule is not applicable, e.g.~we could 
formalize a rule which only matches a node when every incident edge is 
incoming. Additionally the matched occurrences of the patterns can also be 
changed by the rule. A similar approach are adaptive star grammars 
\cite{dhjme:adaptive-star-grammars}, the difference being that we do not 
restrict our left rule sides to be stars.

The resulting formal language can be applied to specify distributed versions 
of concurrent algorithms like Dining Philosophers in which neighbour processes 
use channels to request and grant access to a given shared resource. The 
protocol we use has been proposed by Namjoshi and Trefler in 
\cite{NamjoshiTrefler}. There requests are specified using process identifiers 
attached to edges representing point-to-point communication channels. 
Universally quantified guards are used to ensure mutual exclusive access to a 
resource. In this paper we formulate the protocol without need of introducing 
identifiers. We instead use our extended notion of graph transformation 
systems to specify ownership of a given communication link. Universally 
quantified patterns attached to a requesting node are used then as guards to 
ensure exclusive access. Erroneous or undesirable configurations in the 
algorithm can be presented by a set of minimal error configurations. We then 
use a backward procedure to check if a configuration containing one of the 
error configurations is reachable. If none is reachable, the algorithm is 
proven to be correct.

Following the approach proposed in \cite{BDKSS12,wsts-gts-framework}, we use 
basic ingredients of graph transformation and category theory (e.g. pushouts) 
to formally specify the operational semantics of our model. Parameterized 
verification for the resulting model is undecidable in general, even without 
universally quantified patterns \cite{BDKSS12}. To overcome this problem, we 
provide an approximated symbolic backward procedure using result for 
well-structured transition systems 
\cite{acjt:general-decidability,fs:well-structured-everywhere} to guarantee 
correctness and termination.
\begin{wrapfigure}{r}{6cm}%
\centering%
\scalebox{\ExampleRuleScale}{\begin{tikzpicture}[StdGraphGrid]

\begin{scope}
  \coordinate (lcenter) at (0.5,0.75);
  \node[stdnode] (ln1) at (0.5,0.5) {};
  \node[stdnode] (ln2) at (0,1.5) {};
  \node[stdnode, uqelement] (ln3) at (1,1.5) {};
  \node[henode] (le1) at (0.5,0) {$X$};
  \node[henode, uqelement] (le2) at ($(ln2)!0.5!(ln3)$) {$C$};
  \node[henode, uqelement] (le3) at ($(ln3)!0.5!(ln1)$) {$G$};
  \draw[undiredge] (ln1) -- (le1);
  \draw[graphedge, uqelement] (ln2) -- (le2) -- (ln3);
  \draw[graphedge, uqelement] (ln3) -- (le3) -- (ln1);
\end{scope}
\StaticGraphbox[L]{lcenter}{0.8cm}{1.1cm}{$L$, left}

\begin{scope}[shift={(2.25,0)}]
  \coordinate (rcenter) at (0.5,0.75);
  \node[stdnode] (rn1) at (0.5,0.5) {};
  \node[stdnode] (rn2) at (0,1.5) {};
  \node[stdnode, uqelement] (rn3) at (1,1.5) {};
  \node[henode] (re1) at (0.5,0) {$X$};
  \node[henode, uqelement] (re2) at ($(rn2)!0.5!(rn3)$) {$C$};
  \node[henode, uqelement] (re3) at ($(rn3)!0.5!(rn1)$) {$G$};
  \node[henode] (re4) at ($(rn2)!0.5!(rn1)$) {$G$};
  \draw[undiredge] (rn1) -- (re1);
  \draw[graphedge, uqelement] (rn2) -- (re2) -- (rn3);
  \draw[graphedge, uqelement] (rn3) -- (re3) -- (rn1);
  \draw[graphedge] (rn2) -- (re4) -- (rn1);
\end{scope}
\StaticGraphbox[R]{rcenter}{0.8cm}{1.1cm}{$R$, left}

\draw[mor-parl] (L) -- node [morlabel, above] {$r$} (R);

\end{tikzpicture}%
}
\end{wrapfigure}
Although the over-approximation is based on the monotonic abstraction approach 
proposed in \cite{ADR09,AHDR08}, its application to the considered class of 
infinite-state systems is highly non trivial. In fact, our universal 
quantification approach is not restricted to process states only, but it can 
specify complex graph patterns as shown on the right. There the node marked 
with the $X$-edge represents a group where every node attached with a $G$-edge 
is a member of. The rule can be applied if every edge attached to the two solid 
nodes is matched and has the form of the dashed part (the quantification). 
Effectively the rule adds a node to a group if all other connected nodes (via a 
$C$-edge) are already members of the group.

We have implemented a prototype version of the 
algorithms in the tool \textsc{Uncover} and tested on some case-studies. For 
instance, our prototype can verify the Distributed Dining Philosophers example 
without need of additional invariants as in \cite{NamjoshiTrefler}.
\opt{short}{Due to space limitations, the proofs can be found in an extended 
version of this paper \cite{DS2014arxiv}.}
\opt{long}{Due to space limitations, the proofs can be found in 
Appendix~\ref{app:proofs}.}

\section{Preliminaries}\label{sec:preliminaries}
In this paper we use hypergraphs, a generalization of directed graphs, where 
an edge can connect an arbitrary large but finite set of nodes. Furthermore 
we  use graph morphisms to define rewriting rules.

\paragraph{Hypergraph}
Let $\Lambda$ be a finite sets of edge labels and $\arity \colon \Lambda \to 
\nat$ a function that assigns an arity to each label (including the arity 
zero).  A \emph{($\Lambda$-)hypergraph} (or simply graph) is a tuple $(V_{G}, 
E_{G}, c_{G}, l_G)$ where $V_{G}$ is a finite set of nodes, $E_{G}$ is a 
finite set of edges, $c_{G} \colon E_{G} \rightarrow V_{G}^*$ is a connection 
function and $l_G \colon E_{G} \rightarrow \Lambda$ is an edge labelling 
function. We require that $|c_G(e)| = \arity(l_G(e))$ for each edge $e \in 
E_G$. An edge $e$ is called \emph{incident} to a node $v$ if $v$ occurs in
$c_G(e)$.
An \emph{undirected path} of length $n$ in a hypergraph is an alternating
sequence $v_0, e_1, v_1, \dots, v_{n-1}, e_n, v_n$ of nodes and edges
such that for every index $1 \le i \le n$ both nodes $v_{i-1}$ and
$v_i$ are incident to $e_i$ and the undirected path contains all nodes
and edges at most once.

Let $G$, $G'$ be ($\Lambda$-)hyper\-graphs. A \emph{partial hypergraph 
morphism} (or simply \emph{morphism}) $\phi \colon G\pto G'$ consists of a 
pair of partial functions $(\phi_{V}: V_{G} \pto V_{G'}, \phi_{E}: E_{G} \pto 
E_{G'})$ such that for every $e \in E_{G}$ it holds that $l_{G}(e) = 
l_{G'}(\phi_E(e))$ and $\phi_{V}(c_{G}(e)) = c_{G'}(\phi_{E}(e))$ whenever 
$\phi_E(e)$ is defined. Furthermore if a morphism is defined on an edge, it 
must be defined on all nodes incident to it. We denote \emph{total morphisms} 
by an arrow of the form $\to$ and write $\ito$ if the total morphism is known 
to be injective.

\paragraph{Pushout}
Our rewriting formalism is the so-called \emph{single-pushout approach} (SPO) 
based on the categorical notion of pushouts in the category of graphs and 
partial graph morphisms \cite{handbook-spo}.
Given two morphisms $\phi : G_0 \pto G_1$ and $\psi : G_0 \pto G_2$, the 
\emph{pushout} of $\phi$, $\psi$ consists of the graph $G_3$ and two morphisms 
$\phi' : G_2 \pto G_3$ and $\psi' : G_1 \pto G_3$. It corresponds to a merge 
of $G_1$ and $G_2$ along a common interface $G_0$ while at the same time 
deleting every element of one of the graphs if it has a preimage in $G_0$ 
which is not mapped to an element in the other graph. It is known that in our 
category the pushout of two morphisms always exists and is unique (up to 
isomorphism). It can be computed in the following way.

Let $\equiv_V$ and $\equiv_E$ be the smallest equivalences on $V_{G_1} \cup 
V_{G_2}$ and $E_{G_1} \cup E_{G_2}$ satisfying $\phi(v) \equiv_V \psi(v)$ 
for all $v \in V_{G_0}$ and $\phi(e) \equiv_E \psi(e)$ for all $e \in 
E_{G_0}$. The nodes and edges of the pushout object $G_3$ are then all
\emph{valid} equivalence classes of $\equiv_V$ and $\equiv_E$. An equivalence 
class is \emph{valid} if it does not contain the image of some 
$x \in G_0$ for which $\phi(x)$ or $\psi(x)$ is undefined. The equivalence 
class of an edge is also considered invalid if it is incident to 
a node with an invalid equivalence class. The morphisms $\phi'$ and $\psi'$ 
map each element to its equivalence class if this class is valid and are 
undefined otherwise.

For a backward step in our procedure we also need the notion of a 
\emph{pushout complement} which is, given $\phi : G_0 \pto G_1$ and $\psi' : 
G_1 \pto G_3$, a graph $G_2$ and morphisms $\psi : G_0 \pto G_2$, $\phi' : G_2 
\pto G_3$ such that $G_3$ is the pushout of $\phi$, $\psi$. For graphs pushout 
complements not necessarily exist and if they exist there may be infinitely 
many. See \cite{HJKS:pocs2010} for a detailed description on how pushout 
complements can be computed.

\begin{wrapfigure}{r}{3cm}%
\centering%
\scalebox{1.0}{\begin{tikzpicture}[x=1.5cm,y=-1.5cm]

\node (L) at (0,0) {$L$};
\node (R) at (1,0) {$R$};
\node (G) at (0,1) {$G$};
\node (H) at (1,1) {$H$};
\draw[mor-parl] (L) -- node [midway, above] {$r$} (R);
\draw[mor-tot-inj] (L) -- node [midway, left] {$m$} (G);
\draw[mor-tot-inj] (R) -- node [midway, right] {$m'$} (H);
\draw[mor-parl] (G) -- (H);

\end{tikzpicture}%
}
\end{wrapfigure}

\paragraph{GTS}
A \emph{rewriting rule} is a partial morphism $r\colon L\pto R$, where $L$ is 
called left-hand and $R$ right-hand side. A \emph{match} (of $r$) is a total 
and injective morphism $m : L \ito G$. Given a rule and a match, a 
\emph{rewriting step} or rule application is given by a pushout diagram as 
shown on the right, resulting in the graph $H$.
Note that injective matchings are not a restriction since non-injective 
matchings can be simulated, but are necessary for universally the quantified 
rules defined later.

A \emph{graph transformation system (GTS)} is a finite set of rules
$\mathcal{R}$.  Given a fixed set of graphs $\mathcal{G}$, a
\emph{graph transition system} on $\mathcal{G}$ generated by a graph
transformation system $\mathcal{R}$ is represented by a tuple
$(\mathcal{G},\Rightarrow)$ where $\mathcal{G}$ is the set of states
and $G \Rightarrow G'$ if and only if $G,G' \in \mathcal{G}$ and $G$
can be rewritten to $G'$ using a rule of $\mathcal{R}$.

A computation is a sequence of graphs $G_0, G_1, \ldots$ s.t. $G_i \Rightarrow 
G_{i+1}$ for $i \geq 0$. $G_0$ can reach $G_1$ if there exists a computation 
from $G_0$ to $G_1$.

\section{Graph Transformations with Universally Quantified Conditions}
To clarify the ideas and illustrate the usefulness of universally quantified 
conditions on the neighbourhood of nodes, let us consider the following 
example.

\begin{example}\label{example:main}
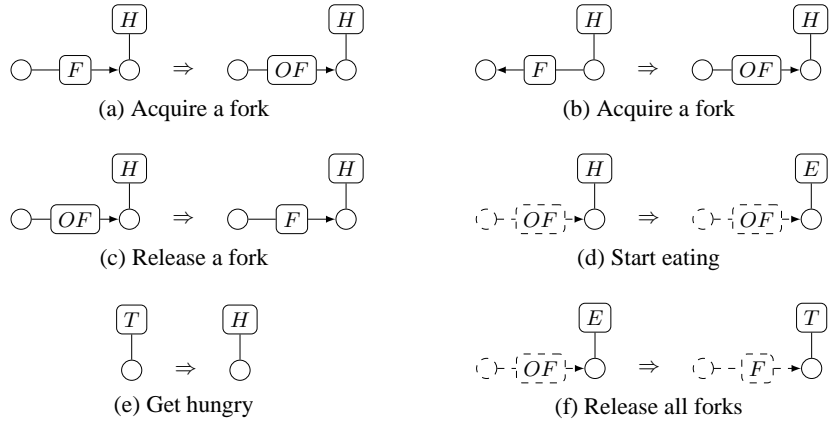
\begin{figure}[ht]
\centering
	\subfloat[Acquire a fork]{%
		\label{fig:main-example-acquire-fork1}
		\begin{minipage}[c]{0.5\textwidth}%
		\centering
		\scalebox{\ExampleRuleScale}{\begin{tikzpicture}[StdGraphGrid]

\begin{scope}
  \node[stdnode] (ln1) at (0,0) {};
  \node[stdnode] (ln2) at (1,0) {};
  \node[henode] (le1) at ($(ln1)!0.5!(ln2)$) {$F$};
  \node[henode] (le2) at (1,-0.5) {$H$};
  \draw[undiredge] (ln2) -- (le2);
  \draw[graphedge] (ln1) -- (le1) -- (ln2);
\end{scope}

\node at (1.5,0) {$\Rightarrow$};

\begin{scope}[shift={(2,0)}]
  \node[stdnode] (rn1) at (0,0) {};
  \node[stdnode] (rn2) at (1,0) {};
  \node[henode] (re1) at ($(rn1)!0.5!(rn2)$) {$OF$};
  \node[henode] (re2) at (1,-0.5) {$H$};
  \draw[undiredge] (rn2) -- (re2);
  \draw[graphedge] (rn1) -- (re1) -- (rn2);
\end{scope}

\end{tikzpicture}%
}
		\end{minipage}}%
	\subfloat[Acquire a fork]{%
		\label{fig:main-example-acquire-fork2}
		\begin{minipage}[c]{0.5\textwidth}%
		\centering
		\scalebox{\ExampleRuleScale}{\begin{tikzpicture}[StdGraphGrid]

\begin{scope}
  \node[stdnode] (ln1) at (0,0) {};
  \node[stdnode] (ln2) at (1,0) {};
  \node[henode] (le1) at ($(ln1)!0.5!(ln2)$) {$F$};
  \node[henode] (le2) at (1,-0.5) {$H$};
  \draw[undiredge] (ln2) -- (le2);
  \draw[graphedge] (ln2) -- (le1) -- (ln1);
\end{scope}

\node at (1.5,0) {$\Rightarrow$};

\begin{scope}[shift={(2,0)}]
  \node[stdnode] (rn1) at (0,0) {};
  \node[stdnode] (rn2) at (1,0) {};
  \node[henode] (re1) at ($(rn1)!0.5!(rn2)$) {$OF$};
  \node[henode] (re2) at (1,-0.5) {$H$};
  \draw[undiredge] (rn2) -- (re2);
  \draw[graphedge] (rn1) -- (re1) -- (rn2);
\end{scope}

\end{tikzpicture}%
}
		\end{minipage}}\\
	\subfloat[Release a fork]{%
		\label{fig:main-example-release-fork}
		\begin{minipage}[c]{0.5\textwidth}%
		\centering
		\scalebox{\ExampleRuleScale}{\begin{tikzpicture}[StdGraphGrid]

\begin{scope}
  \node[stdnode] (ln1) at (0,0) {};
  \node[stdnode] (ln2) at (1,0) {};
  \node[henode] (le1) at ($(ln1)!0.5!(ln2)$) {$OF$};
  \node[henode] (le2) at (1,-0.5) {$H$};
  \draw[undiredge] (ln2) -- (le2);
  \draw[graphedge] (ln1) -- (le1) -- (ln2);
\end{scope}

\node at (1.5,0) {$\Rightarrow$};

\begin{scope}[shift={(2,0)}]
  \node[stdnode] (rn1) at (0,0) {};
  \node[stdnode] (rn2) at (1,0) {};
  \node[henode] (re1) at ($(rn1)!0.5!(rn2)$) {$F$};
  \node[henode] (re2) at (1,-0.5) {$H$};
  \draw[undiredge] (rn2) -- (re2);
  \draw[graphedge] (rn1) -- (re1) -- (rn2);
\end{scope}

\end{tikzpicture}%
}
		\end{minipage}}%
	\subfloat[Start eating]{%
		\label{fig:main-example-start-eating}
		\begin{minipage}[c]{0.5\textwidth}%
		\centering
		\scalebox{\ExampleRuleScale}{\begin{tikzpicture}[StdGraphGrid]

\begin{scope}
  \node[stdnode, uqelement] (ln1) at (0,0) {};
  \node[stdnode] (ln2) at (1,0) {};
  \node[henode, uqelement] (le1) at ($(ln1)!0.5!(ln2)$) {$OF$};
  \node[henode] (le2) at (1,-0.5) {$H$};
  \draw[undiredge] (ln2) -- (le2);
  \draw[graphedge, uqelement] (ln1) -- (le1) -- (ln2);
\end{scope}

\node at (1.5,0) {$\Rightarrow$};

\begin{scope}[shift={(2,0)}]
  \node[stdnode, uqelement] (rn1) at (0,0) {};
  \node[stdnode] (rn2) at (1,0) {};
  \node[henode, uqelement] (re1) at ($(rn1)!0.5!(rn2)$) {$OF$};
  \node[henode] (re2) at (1,-0.5) {$E$};
  \draw[undiredge] (rn2) -- (re2);
  \draw[graphedge, uqelement] (rn1) -- (re1) -- (rn2);
\end{scope}

\end{tikzpicture}%
}
		\end{minipage}}\\
	\subfloat[Get hungry]{%
		\label{fig:main-example-get-hungry}
		\begin{minipage}[c]{0.5\textwidth}%
		\centering
		\scalebox{\ExampleRuleScale}{\begin{tikzpicture}[StdGraphGrid]

\begin{scope}
  \node[stdnode] (ln1) at (0,0) {};
  \node[henode] (le1) at (0,-0.5) {$T$};
  \draw[undiredge] (ln1) -- (le1);
\end{scope}

\node at (0.5,0) {$\Rightarrow$};

\begin{scope}[shift={(1,0)}]
  \node[stdnode] (rn1) at (0,0) {};
  \node[henode] (re1) at (0,-0.5) {$H$};
  \draw[undiredge] (rn1) -- (re1);
\end{scope}

\end{tikzpicture}%
}
		\end{minipage}}%
	\subfloat[Release all forks]{%
		\label{fig:main-example-release-all-forks}
		\begin{minipage}[c]{0.5\textwidth}%
		\centering
		\scalebox{\ExampleRuleScale}{\begin{tikzpicture}[StdGraphGrid]

\begin{scope}
  \node[stdnode, uqelement] (ln1) at (0,0) {};
  \node[stdnode] (ln2) at (1,0) {};
  \node[henode, uqelement] (le1) at ($(ln1)!0.5!(ln2)$) {$OF$};
  \node[henode] (le2) at (1,-0.5) {$E$};
  \draw[undiredge] (ln2) -- (le2);
  \draw[graphedge, uqelement] (ln1) -- (le1) -- (ln2);
\end{scope}

\node at (1.5,0) {$\Rightarrow$};

\begin{scope}[shift={(2,0)}]
  \node[stdnode, uqelement] (rn1) at (0,0) {};
  \node[stdnode] (rn2) at (1,0) {};
  \node[henode, uqelement] (re1) at ($(rn1)!0.5!(rn2)$) {$F$};
  \node[henode] (re2) at (1,-0.5) {$T$};
  \draw[undiredge] (rn2) -- (re2);
  \draw[graphedge, uqelement] (rn1) -- (re1) -- (rn2);
\end{scope}

\end{tikzpicture}%
}
		\end{minipage}}%
\caption{Modelling of the dining philosophers problem on an arbitrary net}
\label{fig:main-example}
\end{figure}

Figure~\ref{fig:main-example} shows a set of rules describing the Dining 
Philosophers Problem on an arbitrary graph structure. Each node represents a 
philosopher who can be in one of three different states: hungry ($H$), eating 
($E$) or thinking ($T$). Each state is indicated by a unary edge attached to 
the philosopher. Between two philosophers there may be a free fork (an 
$F$-edge) or a fork owned by one of the philosophers (an $OF$-edge pointing to 
its owner). Note that our directed edges are in fact hyperedges of arity two, 
where the first node is the source and the second node is the target.

Philosophers can take unowned forks 
(Figure~\ref{fig:main-example-acquire-fork1} 
and~\ref{fig:main-example-acquire-fork2}) and also release control 
(Figure~\ref{fig:main-example-release-fork}). 
If a philosopher owns all connected forks, he can start to eat 
(Figure~\ref{fig:main-example-start-eating}). The dashed part of the rule 
indicates a universal quantification, meaning that the rule can only be applied 
if all edges attached to the philosopher are part of the matching and in fact 
forks owned by him. 
At some point the philosopher finished eating, releasing all forks 
(Figure~\ref{fig:main-example-release-all-forks}) and may become 
hungry in the future (Figure~\ref{fig:main-example-get-hungry}). 
When releasing all forks, all forks owned by the philosopher are converted to 
unowned forks.

Rules matching the entire neighbourhood of a node (in the following called 
\emph{quantified node}), such as the rules in 
Figure~\ref{fig:main-example-start-eating} 
and~\ref{fig:main-example-release-all-forks} cannot be described by normal 
rewriting rules. Therefore we extend normal rules to so-called universally 
quantified rules consisting of a normal rule and a set of universal 
quantifications. The idea is to first find a matching for the rule and then 
extend the rule as well as the matching until the entire neighbourhood of 
quantified nodes is part of the matching.

We apply the rule in Figure~\ref{fig:main-example-release-all-forks} to the 
graph $G$ shown in Figure~\ref{fig:broadcast-rule-motivation}. There exists a 
match $m : L \ito G$ where $r : L \pto R$ is the rule without any use of the 
quantification. However, this matching does not match the entire neighbourhood 
of the quantified node (marked grey). Before applying the rule we have to add 
multiple copies of the quantification to $r$ generating a so-called 
instantiation $\eta$ where the extended match $\overline{m}$ contains the entire 
neighbourhood of the quantified node.

\begin{figure}[ht]
  \centering
  \scalebox{\ExampleRuleScale}{\begin{tikzpicture}[StdGraphGrid]

\begin{scope}
  \coordinate (lcenter) at (0,0);
  \node[stdnode, uqnode] (ln2) at (0,0.25) {};
  \node[henode] (le2) at (0,-0.25) {$E$};
  \draw[undiredge] (ln2) -- (le2);
\end{scope}
\StaticGraphbox[L]{lcenter}{0.4cm}{0.4cm}{$L$, left}

\begin{scope}[shift={(2.2,0)}]
  \coordinate (rcenter) at (0,0);
  \node[stdnode, uqnode] (rn2) at (0,0.25) {};
  \node[henode] (re2) at (0,-0.25) {$T$};
  \draw[undiredge] (rn2) -- (re2);
\end{scope}
\StaticGraphbox[R]{rcenter}{0.4cm}{0.4cm}{$R$, left}

\begin{scope}[shift={(0,1.8)}]
  \coordinate (olcenter) at (0,0);
  \node[stdnode] (oln1) at (-0.5,-0.5) {};
  \node[stdnode, uqnode] (oln2) at (0.5,0) {};
  \node[stdnode] (oln3) at (-0.5,0.5) {};
  \node[henode] (ole2) at (0.5,-0.5) {$E$};
  \node[henode] (ole12) at ($(oln1)!0.5!(oln2)$) {$OF$};
  \node[henode] (ole32) at ($(oln3)!0.5!(oln2)$) {$OF$};
  \draw[undiredge] (oln2) -- (ole2);
  \draw[graphedge] (oln1) -- (ole12) -- (oln2);
  \draw[graphedge] (oln3) -- (ole32) -- (oln2);
\end{scope}
\StaticGraphbox[oL]{olcenter}{0.8cm}{0.8cm}{}

\begin{scope}[shift={(2.2,1.8)}]
  \coordinate (orcenter) at (0,0);
  \node[stdnode] (orn1) at (-0.5,-0.5) {};
  \node[stdnode, uqnode] (orn2) at (0.5,0) {};
  \node[stdnode] (orn3) at (-0.5,0.5) {};
  \node[henode] (ore2) at (0.5,-0.5) {$T$};
  \node[henode] (ore12) at ($(orn1)!0.5!(orn2)$) {$F$};
  \node[henode] (ore32) at ($(orn3)!0.5!(orn2)$) {$F$};
  \draw[undiredge] (orn2) -- (ore2);
  \draw[graphedge] (orn1) -- (ore12) -- (orn2);
  \draw[graphedge] (orn3) -- (ore32) -- (orn2);
\end{scope}
\StaticGraphbox[oR]{orcenter}{0.8cm}{0.8cm}{}

\begin{scope}[shift={(-2.5,0.9)}]
  \coordinate (gcenter) at (0,0);
  \node[stdnode] (gn1) at (-0.25,-0.25) {};
  \node[stdnode, uqnode] (gn2) at (0.75,0.25) {};
  \node[stdnode] (gn3) at (-0.25,0.75) {};
  \node[henode] (ge1) at (-0.25,-0.75) {$H$};
  \node[henode] (ge2) at (0.75,-0.25) {$E$};
  \node[henode] (ge3) at (-0.25,0.25) {$H$};
  \node[henode] (ge12) at ($(gn1)!0.5!(gn2)$) {$OF$};
  \node[henode] (ge32) at ($(gn3)!0.5!(gn2)$) {$OF$};
  \node[henode] (ge13) at ($(gn1)!0.5!(gn3) + (-0.5,0)$) {$F$};
  \draw[undiredge] (gn1) -- (ge1);
  \draw[undiredge] (gn2) -- (ge2);
  \draw[undiredge] (gn3) -- (ge3);
  \draw[graphedge] (gn1) -- (ge12) -- (gn2);
  \draw[graphedge] (gn3) -- (ge32) -- (gn2);
  \draw[graphedge] (gn1) -- (ge13) -- (gn3);
\end{scope}
\StaticGraphbox[G]{gcenter}{1.2cm}{1.2cm}{$G$, left}

\draw[mor-parl] (L) -- node [morlabel, above] {$r$} (R);
\draw[mor-tot-inj, dash pattern= on 1pt off 1pt] (L) -- (oL);
\draw[mor-parl] (oL) -- node [morlabel, above] {$\eta$} (oR);
\draw[mor-tot-inj, dash pattern= on 1pt off 1pt] (R) -- (oR);
\draw[mor-tot-inj] (L) -- node [morlabel, above left] {$m$} (G);
\draw[mor-tot-inj] (oL) -- node [morlabel, above right] {$\overline{m}$} (G);

\end{tikzpicture}%
}
  \caption{A match of a universally quantified rule has to be extended until 
  the entire neighbourhood of each quantified node is matched}
  \label{fig:broadcast-rule-motivation}
\end{figure}
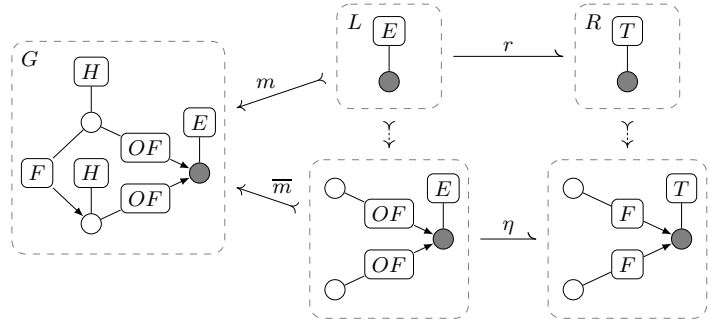
\end{example}

In the following we formalize the notion of universally quantified rules as an 
extension of normal rules and introduce instantiations via a sequence of 
recursive instantiation steps.

\begin{definition}[Universally quantified rules]\label{def:uqrule}
A \emph{universally quantified rule} is a pair $\rho = (r,U)$, where $r : L 
\pto R$ is a partial morphism and $U$ is a finite set of universal 
quantifications. A \emph{universal quantification} is a pair $(p_u,q_u) = u 
\in U$ where $p_u : L \ito L_u$ is a total injective morphism and $q_u : L_u 
\pto R_u$ is a partial morphism satisfying the restriction that $q_u(p_u(x))$ 
is defined and has exactly one preimage in $L_u$ for every $x \in L$.

With $\qNodes(u)$ we denote the set of \emph{quantified nodes} of $u$, which 
is the set of all $v \in V_L$ such that there is an edge incident to $p_u(v)$ 
which has no preimage in $L$. We denote the quantified nodes of a rule the 
same way, i.e.~$\qNodes(\rho) = \bigcup_{u \in U}\qNodes(u)$.
We require that $\qNodes(u) \neq \emptyset$ for all $u \in U$.
\end{definition}

In the rest of the paper we will use UGTS to denote 
the 
extension of GTS with universally quantified rules.

\begin{definition}[Instantiation of a universally quantified rule]
\label{def:uqrule-instantiation}
An instantiation of a universally quantified rule $\rho = (r, U)$ consists 
of a total injective morphism $\pi : L \ito \overline{L}$ and a partial 
morphism $\gamma: \overline{L} \pto \overline{R}$ and is recursively defined as 
follows:

\noindent
\parbox{0.66\textwidth}{%
\begin{itemize}
  \item The pair $(\mathit{id}_L : L \ito L, r)$, where $\mathit{id}_L$ is 
  the identity on $L$, is an instantiation of $\rho$.
  \item Let $(\pi : L \ito \overline{L}, \gamma : \overline{L} \pto 
  \overline{R})$ be an instantiation of $\rho$ and let $(p_u : L \ito L_u, q_u 
  : L_u \pto R_u) = u \in U$. Furthermore, let $\overline{L}_u$ be the pushout 
  of $\pi$, $p_u$ and let $\overline{R}_u$ be the pushout of $\gamma \circ 
  \pi$, $q_u \circ p_u$, as shown in the diagram to the right. Then $p_u' 
  \circ \pi$ and the (unique) mediating morphism $\eta$ are also an 
  instantiation of $\rho$. We write $(p_u' \circ \pi, \eta) = (\pi, 
  \gamma) \instAdd u$ to indicate that the instantiation $(\pi, \gamma)$ was 
  extended by $u$.
\end{itemize}}
\parbox{0.34\textwidth}{%
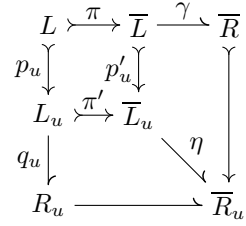
\begin{center}
  \scalebox{1.0}{\begin{tikzpicture}[x=1.2cm,y=-1.2cm]

\node (L) at (0,0) {$L$};
\node (oL) at (1,0) {$\overline{L}$};
\node (oR) at (2,0) {$\overline{R}$};
\node (Li) at (0,1) {$L_u$};
\node (oLi) at (1,1) {$\overline{L}_u$};
\node (Ri) at (0,2) {$R_u$};
\node (oRi) at (2,2) {$\overline{R}_u$};

\draw[mor-tot-inj] (L) -- node [morlabel, above] {$\pi$} (oL);
\draw[mor-parl] (oL) -- node [morlabel, above] {$\gamma$} (oR);
\draw[mor-tot-inj] (L) -- node [morlabel, left] {$p_u$} (Li);
\draw[mor-tot-inj] (oL) -- node [morlabel, left] {$p_u'$} (oLi);
\draw[mor-tot-inj] (Li) -- node [morlabel, above] {$\pi'$} (oLi);
\draw[mor-parl] (Li) -- node [morlabel, left] {$q_u$} (Ri);
\draw[mor-parl] (oLi) -- node [morlabel, above right] {$\eta$} (oRi);
\draw[mor-tot-inj] (oR) -- (oRi);
\draw[mor-parl] (Ri) -- (oRi);

\end{tikzpicture}%
}
\end{center}}
We say that the length of an instantiation is the number of steps performed 
to generate the instantiation, where $(\mathit{id}_L,r)$ has a length of $0$.
\end{definition}

\begin{example}
Figure~\ref{fig:broadcast-rule-instantiation-example} shows a possible 
instantiation of the rule in Figure~\ref{fig:main-example-release-all-forks}. 
There is only one universal quantification $u$ and this quantification is used 
once to generate the instantiation $(p_u' \circ id_L, \eta)$. Any further 
instantiation will add an additional node and $OF$-edge to $\overline{L}_u$ and 
an additional node and $F$-edge to $\overline{R}_u$. The universally quantified 
node (i.e.~$\qNodes(u)$) is marked grey. This means that $\eta$ is only 
applicable if the grey node is matched to a node with degree (exactly) two. The 
rule application is performed by calculating the pushout of $\eta$ (not $r$) 
and a valid matching $m$. The matching is only valid if all edges 
incident to the grey node have a preimage in $\overline{L}_u$, such that an 
application will always result in all incident $OF$-edges to be replaced by 
$F$-edges. Although the number of affected edges can be arbitrary large, the 
quantification it bounded to the neighbourhood of the grey node and therefore 
the change is still local.

\begin{figure}[ht]
  \centering
  \scalebox{\ExampleRuleScale}{\begin{tikzpicture}[StdGraphGrid]

\begin{scope}
  \coordinate (lcenter) at (0,0);
  \node[stdnode, uqnode] (ln2) at (0,0.25) {};
  \node[henode] (le2) at (0,-0.25) {$E$};
  \draw[undiredge] (ln2) -- (le2);
\end{scope}
\StaticGraphbox[L]{lcenter}{0.4cm}{0.4cm}{$L$, left}

\begin{scope}[shift={(2,0)}]
  \coordinate (olcenter) at (0,0);
  \node[stdnode, uqnode] (oln2) at (0,0.25) {};
  \node[henode] (ole2) at (0,-0.25) {$E$};
  \draw[undiredge] (oln2) -- (ole2);
\end{scope}
\StaticGraphbox[oL]{olcenter}{0.4cm}{0.4cm}{$L$, left}

\begin{scope}[shift={(4,0)}]
  \coordinate (orcenter) at (0,0);
  \node[stdnode, uqnode] (orn2) at (0,0.25) {};
  \node[henode] (ore2) at (0,-0.25) {$T$};
  \draw[undiredge] (orn2) -- (ore2);
\end{scope}
\StaticGraphbox[oR]{orcenter}{0.4cm}{0.4cm}{$R$, left}

\begin{scope}[shift={(0,1.5)}]
  \coordinate (lucenter) at (0,0);
  \node[stdnode] (lun1) at (-0.5,0.25) {};
  \node[stdnode, uqnode] (lun2) at (0.5,0.25) {};
  \node[henode] (lue1) at ($(lun1)!0.5!(lun2)$) {$OF$};
  \node[henode] (lue2) at (0.5,-0.25) {$E$};
  \draw[undiredge] (lun2) -- (lue2);
  \draw[graphedge] (lun1) -- (lue1) -- (lun2);
\end{scope}
\StaticGraphbox[Lu]{lucenter}{0.8cm}{0.4cm}{$L_u$, left}

\begin{scope}[shift={(2,1.5)}]
  \coordinate (olucenter) at (0,0);
  \node[stdnode] (olun1) at (-0.5,0.25) {};
  \node[stdnode, uqnode] (olun2) at (0.5,0.25) {};
  \node[henode] (olue1) at ($(olun1)!0.5!(olun2)$) {$OF$};
  \node[henode] (olue2) at (0.5,-0.25) {$E$};
  \draw[undiredge] (olun2) -- (olue2);
  \draw[graphedge] (olun1) -- (olue1) -- (olun2);
\end{scope}
\StaticGraphbox[oLu]{olucenter}{0.8cm}{0.4cm}{$\overline{L}_u$, left}

\begin{scope}[shift={(0,3.0)}]
  \coordinate (rucenter) at (0,0);
  \node[stdnode] (run1) at (-0.5,0.25) {};
  \node[stdnode, uqnode] (run2) at (0.5,0.25) {};
  \node[henode] (rue1) at ($(run1)!0.5!(run2)$) {$F$};
  \node[henode] (rue2) at (0.5,-0.25) {$E$};
  \draw[undiredge] (run2) -- (rue2);
  \draw[graphedge] (run1) -- (rue1) -- (run2);
\end{scope}
\StaticGraphbox[Ru]{rucenter}{0.8cm}{0.4cm}{$R_u$, left}

\begin{scope}[shift={(4,3.0)}]
  \coordinate (orucenter) at (0,0);
  \node[stdnode] (orun1) at (-0.5,0.25) {};
  \node[stdnode, uqnode] (orun2) at (0.5,0.25) {};
  \node[henode] (orue1) at ($(orun1)!0.5!(orun2)$) {$F$};
  \node[henode] (orue2) at (0.5,-0.25) {$T$};
  \draw[undiredge] (orun2) -- (orue2);
  \draw[graphedge] (orun1) -- (orue1) -- (orun2);
\end{scope}

\StaticGraphbox[oRu]{orucenter}{0.8cm}{0.4cm}{$\overline{R}_u$, left}

\draw[mor-tot-inj] (L) -- node [morlabel, above] {$\mathit{id}_L$} (oL);
\draw[mor-parl] (oL) -- node [morlabel, above] {$r$} (oR);
\draw[mor-tot-inj] (L) -- node [morlabel, left] {$p_u$} (Lu);
\draw[mor-tot-inj] (oL) -- node [morlabel, left] {$p_u'$} (oLu);
\draw[mor-tot-inj] (Lu) -- node [morlabel, above, inner sep=3pt] 
{$\mathit{id}_L'$} (oLu);
\draw[mor-parl] (Lu) -- node [morlabel, left] {$q_u$} (Ru);
\draw[mor-parl] (oLu) -- node [morlabel, above right] {$\eta$} (oRu);
\draw[mor-tot-inj] (oR) -- (oRu);
\draw[mor-parl] (Ru) -- (oRu);

\end{tikzpicture}%
}
  \caption{A possible instantiation of the rule in 
  Figure~\ref{fig:main-example-release-all-forks}}
  \label{fig:broadcast-rule-instantiation-example}
\end{figure}
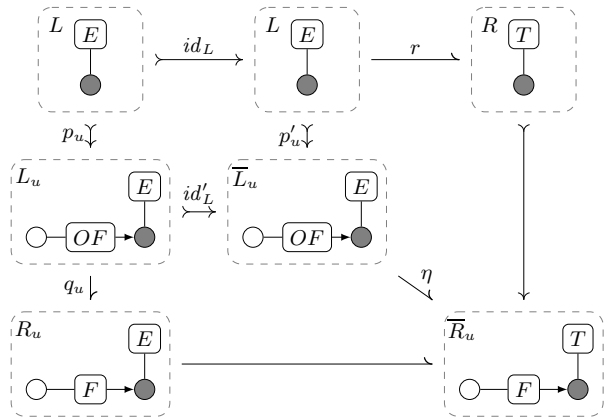
\end{example}

The order in which universal quantifications are used to generate 
instantiations can be neglected, since different sequences will still yield the 
same instantiation (up to isomorphism). Therefore we can uniquely specify 
instantiations by the number each universal quantification in its sequence.

\begin{definition}[Rule application]\label{def:uqrule-application}
Let $\rho$ be a universally quantified rule. We say that $\rho$ is applicable 
to a graph $G$, if there is an instantiation $(\pi, \gamma)$ of $\rho$ and a 
total injective match $m : \overline{L} \ito G$, such that for every $x \in 
\qNodes(\rho)$, there is no $e \in E_G$ incident to $m(\pi(x))$ without a 
preimage in $\overline{L}$. The application of $\rho$ to $G$ via $m$ results in 
the graph $H$, the pushout of $m$ and $\gamma$.
\end{definition}
We reuse the notation $G \Rightarrow G'$ to denote a rewriting step from $G$ to $G'$.
The previous definition introduces a restricted form of negative application 
condition since the existence of an edge, which cannot be mapped by a 
quantification, may block the application of a rule. 

\section{A Procedure for Coverability in UGTS}
In this paper we focus our attention on verification problems that can be 
formulated as reachability and coverability decision problems. Given an initial 
configuration $G_0$ and a target configuration $G_1$ reachability consists in 
checking whether there exists a computation from $G_0$ to $G_1$. The 
coverability problem is similar to the reachability problem, but additionally 
relies on an ordering. In this paper we use the subgraph ordering, but there 
are other suitable orders such as the minor ordering or the induced subgraph 
ordering \cite{wsts-gts-framework}.

\begin{definition}[Subgraph Ordering]
A graph $G_1$ is a subgraph of $G_2$, written $G_1 \subOrder G_2$, if there 
exists a partial, injective and surjective morphism from $G_2$ to $G_1$, 
written $\mu : G_2 \subArrow G_1$. Such morphisms are called subgraph morphisms.
\end{definition}

Given a $G$, a subgraph can always be obtained by a sequence of node and edge 
deletions. Note that due to the morphism property every edge attached to a 
deleted node must be deleted as well. Using the subgraph ordering we can 
represent sets of configurations by minimal graphs and define two variants of 
the coverability problem.

\begin{definition}[Upward Closure]
The \emph{upward closure} of a set $\mathcal{S}$ of graphs is defined as 
$\upclosed{\mathcal{S}} = \{G' \mid G \subOrder G', G \in \mathcal{S}\}$.
A set $\mathcal{S}$ is \emph{upward-closed} if it satisfies $\mathcal{S} = 
\upclosed{\mathcal{S}}$. A \emph{basis} of an upward-closed set $\mathcal{S}$ 
is a set $\mathcal{B}$ such that $\mathcal{S} = \upclosed{\mathcal{B}}$.
\end{definition}

\begin{definition}[Coverability]
Let $G_0$, $G_1$ be two graphs. The \emph{general coverability problem} is to 
decide whether from $G_0$ we can reach a graph $G_2$ such that $G_1 
\subOrder G_2$.

Let $\mathcal{G}$ a set of graphs and let $G_0, G_1 \in \mathcal{G}$. The 
\emph{restricted coverability problem} is to decide whether from $G_0$ we can 
reach a graph $G_2 \in \mathcal{G}$ such that $G_1 \subOrder G_2$ and 
every graph on the sequence from $G_0$ to $G_2$ is an element of $\mathcal{G}$.
\end{definition}

In other words, a configuration is coverable from some initial configuration if 
we can reach a configuration containing (as subgraph) a given pattern.
Although general and restricted coverability are both undecidable, we can 
obtain decidability results by using a backward search introduced for 
well-structured transition systems 
\cite{acjt:general-decidability,fs:well-structured-everywhere} as already shown 
in \cite{BDKSS12}. These systems rely on a \emph{well-quasi-order (wqo)}, which 
is a transitive reflexive order $\leq$ such that there is no infinite, strictly 
decreasing sequence of elements and no infinite antichain, a sequence of 
pairwise incomparable elements, wrt.~$\leq$. A direct consequence of this 
property is that every upward-closed set wrt.~some wqo has a finite basis. It 
has been shown that the subgraph ordering is a well-quasi-order on 
$\boundedPathk$, the class of graphs in which every  undirected path has at 
most the length $k$ \cite{d:subgraphs-wqo}. We remark that the property does 
not hold if only directed paths are restricted.

The backward search presented in this paper is a version of the general 
backward search presented in \cite{wsts-gts-framework} adapted to be compatible 
with UGTS. We denote the set of \emph{predecessors} for a set of graphs 
$\mathcal{S}$ by $\pred{\mathcal{S}} = \{G' \mid \exists G \in \mathcal{S} 
\colon G' \Rightarrow G\}$. Furthermore we denote the predecessors reachable 
within multiple step by \predAll{\mathcal{S}} and the \emph{restricted 
predecessors} by $\pred[\mathcal{G}]{\mathcal{S}} = \pred{\mathcal{S}} \cap 
\mathcal{G}$. We will present a procedure for UGTS to compute so-called 
effective pred-basis and effective $\boundedPathk$-pred-basis. An 
\emph{effective pred-basis} for a graph $G$ is a finite basis $\predBasis{G}$ 
of $\upclosed{\pred{\upclosed{\{G\}}}}$ and an \emph{effective 
$\boundedPathk$-pred-basis} is a finite basis $\predBasis[k]{G}$ of 
$\upclosed{\pred[\boundedPathk]{\upclosed{\{G\}}}}$. Using the effective 
$\boundedPathk$-pred-basis the backward search will terminate and compute a 
finite basis $\mathcal{B}$. If $G \in \upclosed{\mathcal{B}}$, then $G$ covers 
a configuration of $\mathcal{S}$ in $\Rightarrow$ (general coverability). If $G 
\notin \upclosed{\mathcal{B}}$, then $G$ does not cover a configuration of 
$\mathcal{S}$ in $\Rightarrow_{\boundedPathk}$ (no restricted coverability), 
where $\Rightarrow_{\boundedPathk}$ is the restriction $\Rightarrow \cap\ 
(\boundedPathk \times \boundedPathk)$. By using the effective pred-basis the 
backward search computes a finite basis for \predAll{\mathcal{S}}, but is not 
guaranteed to terminate.

The computation of a $\boundedPathk$-pred-basis is performed by 
Procedure~\ref{procedure:back-step}. We assume that for a graph $G$ and a rule 
$\rho$ there is an upper bound on the length of instantiations necessary to 
compute a backward step and write \instBound[\rho]{G} to denote 
such an upper bound. The existence of this upper bound is shown later on in 
Proposition~\ref{prop:inst-bound-exists}. The result of a backward step is a 
finite set $\mathcal{S}$ of graphs such that $\pred{\upclosed{\{G\}}} \subseteq 
\upclosed{\mathcal{S}}$.

\begin{procedure}[Backward Step]\label{procedure:back-step}
\mbox{}\\%
\pseudoParagraph{Input}
A rule  $\rho$ and a graph $G$.

\pseudoParagraph{Procedure}
\begin{compactenum}
  \item First compute all instantiations $(\pi : L \ito \overline{L}, \gamma 
  : \overline{L} \pto \overline{R})$ of $\rho$ up to the length 
  $\instBound[\rho]{G}$.
  \item \label{backstep:concat-subgraphs}
  For each $\gamma$ compute all subgraph morphisms $\mu : \overline{R} 
  \subArrow R'$. Note that it is sufficient to take a representative $R'$ for 
  each of the finitely many isomorphism classes.
  \item For each $\mu \circ \gamma$ compute all total injective morphisms  
  $m' : R' \to G$ (co-matches of $R'$ in $G$).
  \item \label{backstep:calculate-pocs}
  For each such morphism $m'$ calculate all minimal pushout complements $G'$, 
  $m : \overline{L} \ito G'$ of $m'$ and $\mu \circ  \gamma$ where $m$ is 
  injective and $G'$ is an element of~$\boundedPathk$. 
  Drop all $G'$ where $m$ does not satisfy the application condition of 
  Definition~\ref{def:uqrule-application}, i.e.~there is an edge incident to a 
  quantified node which is not in the matching.
\end{compactenum}

\pseudoParagraph{Result}
The set of all graphs not dropped in Step~\ref{backstep:calculate-pocs}, 
written \predBasis[k]{G}. 
\end{procedure}%

The motivation behind Step~\ref{backstep:concat-subgraphs} is that $G$ 
represents not just itself but also its upward closure. Therefore, the rule 
must also be applied to every graph larger than $G$. Instead of using partial 
co-matches we concatenate with subgraph morphisms to simulate this behaviour.

The procedure for a single backward step can be used to define a backward 
search procedure for the coverability problem for UGTS. The procedure exploits 
the property that, even if compatibility is not satisfied, 
$\pred{\upclosed{\mathcal{S}}} \subseteq 
\upclosed{\pred{\upclosed{\mathcal{S}}}}$ still holds for every set of graphs 
$\mathcal{S}$. We can iteratively compute backward steps for all minimal graphs 
$G$ of $\upclosed{\mathcal{S}}$ and check that no initial state is reached 
backwards. 

\begin{procedure}[Backward Search]\label{procedure:main}
\mbox{}\\%
\pseudoParagraph{Input} 
A natural number $k$, a set $\mathcal{R}$ of graph transformation rules and a 
finite set of final graphs $\mathcal{F}$. Start with the working set 
$\mathcal{W} = \mathcal{F}$.

\pseudoParagraph{Backward Step} For each $G \in \mathcal{W}$ add all graphs of  
\predBasis[k]{G} to $\mathcal{W}$ and minimize $\mathcal{W}$ by removing all 
graphs $H'$ for which there is a graph $H'' \in \mathcal{W}$ with $H' \neq H''$ 
and $H''\subOrder H'$. Repeat this backward steps until the sequence of  
working sets $\mathcal{W}$ becomes stationary, i.e.~for every $G \in 
\mathcal{W}$ the computation of the backward step using $G$ results in no 
change of $\mathcal{W}$. 

\pseudoParagraph{Result} The resulting set $\mathcal{W}$ contains minimal 
representatives of graphs from which a final state is coverable. This set may 
be an over-approximation, even without quantified rules.
\end{procedure}

To show the termination of Procedure~\ref{procedure:back-step} 
and~\ref{procedure:main} it is important to show the existence of a bounding 
function \instBound[\rho]{}. By the following proposition this function exists 
for every rule $\rho$, but as we will show later this bound can be tightened in 
most cases.

\newcommand{\propInstBoundExists}{%
Let $\iota$ be an instantiation of length $k$ of some rule $\rho$. If $k$ is 
larger than the number of nodes and edges of $G$, then every graph computed 
by the backward application of $\iota$ is already represented by the backward 
application of an instantiation of lower length.}
\begin{proposition}\label{prop:inst-bound-exists}
\propInstBoundExists
\end{proposition}

The following two lemmas prove that Procedure~\ref{procedure:back-step} 
computes a finite basis of an over-approximation of the restricted predecessors.

\newcommand{\lemmaProcCorrectnessOne}{%
The set \predBasis[k]{G} is a finite subset of 
\pred{\upclosed{\{G\}}} and $\predBasis[k]{G} \subseteq \boundedPathk$.}
\begin{lemma}\label{lem:proc-correctness1}
\lemmaProcCorrectnessOne
\end{lemma}

\newcommand{\lemmaProcCorrectnessTwo}{%
It holds that $\upclosed{\predBasis[k]{G}} \supseteq 
\upclosed{\pred[\boundedPathk]{\upclosed{\{G\}}}}$.}
\begin{lemma}\label{lem:proc-correctness2}
\lemmaProcCorrectnessTwo
\end{lemma}

We recapitulate our main result in the following proposition.

\begin{proposition}\label{prop:proc-correctness}
For each graph $G$, \predBasis[k]{G} is an effective 
$\boundedPathk$-pred-basis. Furthermore, Procedure~\ref{procedure:main} 
terminates and computes an over-approximation of all configurations in 
$\boundedPathk$ from which a final configuration is coverable.
\end{proposition}

\begin{proof}
By Lemma~\ref{lem:proc-correctness1} and~\ref{lem:proc-correctness2} we know 
that $\upclosed{\predBasis[k]{G}} = 
\upclosed{\pred[\boundedPathk]{\upclosed{\{G\}}}}$ and thus \predBasis[k]{G} is 
a $\boundedPathk$-pred-basis.
According to Proposition~\ref{prop:inst-bound-exists} for every $\rho \in 
\mathcal{R}$ the number of necessary instantiation steps is 
bounded by $\instBound[\rho]{G}$, thus, the number of instantiations is fine. 
For each instantiation the minimal pushout complements restricted to 
$\boundedPathk$ are finite and computable. Since the subgraph ordering is 
decidable the minimization is computable and \predBasis[k]{G} is effective.

Since the subgraph ordering is a wqo on $\boundedPathk$, every infinite 
increasing sequence of upward-closed set becomes stationary. The 
upward-closures of the working sets $\mathcal{W}$ form such an infinite 
increasing sequence, thus the termination criteria of 
Procedure~\ref{procedure:main} will be satisfied at some point. \qed
\end{proof}

\subsection*{A Variant of \predBasis[k]{} Without Path Bound}
In Step~\ref{backstep:calculate-pocs} of Procedure~\ref{procedure:back-step} 
every graph which is not an element of $\boundedPathk$ is dropped. This is 
needed to guarantee that the working set of Procedure~\ref{procedure:main} 
becomes stationary and the search terminates. However, this restriction can be 
dropped to obtain a backward search which solves the general coverability 
problem. Termination is not guaranteed, but correctness can be proven 
analogously to the restricted variant, as already shown in 
\cite{wsts-gts-framework}. Let \predBasis{} be 
Procedure~\ref{procedure:back-step} without the restriction to $\boundedPathk$. 
We summarize the decidability of this second variant in the following 
proposition.

\begin{proposition}\label{prop:proc-correctness-unrestricted}
For each graph $G$, \predBasis{G} is an effective pred-basis. Furthermore, when 
using \predBasis{} instead of \predBasis[k]{}, Procedure~\ref{procedure:main} 
computes an over-approximation of all configurations from which a final 
configuration is coverable.
\end{proposition}

\subsection*{Experimental Results}
We added support for universally quantified rules to the \textsc{Uncover} tool. 
This tool can perform the backward search for the subgraph ordering and the 
minor ordering (a coarser order compared to subgraphs). Both variants of the 
backward search are implemented, but a timeout might occur when using the 
unresticted variant. However, given the rules in Figure~\ref{fig:main-example} 
and the error graphs in Figure~\ref{fig:main-example-errors} the unrestricted 
variant terminates after 12 seconds and results in a set of 12 minimal graphs. 
Two of these graphs are the initial error graphs and two other computed graphs 
are shown in Figure~\ref{fig:main-example-another-errors}. Every minimal graph 
contains a node in the state $E$. Since initially no philosopher is eating, the 
initial configuration is not represented and none of the initial error graphs 
is reachable. This proves that two adjacent philosophers cannot be eating at the
same time.

\begin{figure}[ht]
\begin{minipage}[c]{0.45\textwidth}
  \subfloat{%
    \label{fig:main-example-error1}
    \begin{minipage}[c]{0.5\textwidth}%
    \centering
    \scalebox{\ExampleRuleScale}{\begin{tikzpicture}[StdGraphGrid]

\node[stdnode] (n1) at (0,0) {};
\node[stdnode] (n2) at (1,0) {};
\node[henode] (e1) at (0,-0.5) {$E$};
\node[henode] (e2) at (1,-0.5) {$E$};
\node[henode] (e3) at ($(n1)!0.5!(n2)$) {$F$};
\draw[undiredge] (n1) -- (e1);
\draw[undiredge] (n2) -- (e2);
\draw[graphedge] (n1) -- (e3) -- (n2);

\end{tikzpicture}%
}
  \end{minipage}}%
  \subfloat{%
    \label{fig:main-example-error2}
    \begin{minipage}[c]{0.5\textwidth}%
    \centering
    \scalebox{\ExampleRuleScale}{\begin{tikzpicture}[StdGraphGrid]

\node[stdnode] (n1) at (0,0) {};
\node[stdnode] (n2) at (1,0) {};
\node[henode] (e1) at (0,-0.5) {$E$};
\node[henode] (e2) at (1,-0.5) {$E$};
\node[henode] (e3) at ($(n1)!0.5!(n2)$) {$OF$};
\draw[undiredge] (n1) -- (e1);
\draw[undiredge] (n2) -- (e2);
\draw[graphedge] (n1) -- (e3) -- (n2);

\end{tikzpicture}%
}
  \end{minipage}}%
\caption{Two error configurations in the Dining Philosophers Problem}
\label{fig:main-example-errors}
\end{minipage}\hfill
\begin{minipage}[c]{0.45\textwidth}
  \subfloat{%
    \label{fig:main-example-another-error1}
    \begin{minipage}[c]{0.5\textwidth}%
    \centering
    \scalebox{\ExampleRuleScale}{\begin{tikzpicture}[StdGraphGrid]

\node[stdnode] (n1) at (0,0) {};
\node[stdnode] (n2) at (1,0) {};
\node[henode] (e1) at (0,-0.5) {$H$};
\node[henode] (e2) at (1,-0.5) {$E$};
\node[henode] (e3) at ($(n1)!0.5!(n2)$) {$F$};
\draw[undiredge] (n1) -- (e1);
\draw[undiredge] (n2) -- (e2);
\draw[graphedge] (n1) -- (e3) -- (n2);

\end{tikzpicture}%
}
  \end{minipage}}%
  \subfloat{%
    \label{fig:main-example-another-error2}
    \begin{minipage}[c]{0.5\textwidth}%
    \centering
    \scalebox{\ExampleRuleScale}{\begin{tikzpicture}[StdGraphGrid]

\node[stdnode] (n1) at (0,0) {};
\node[stdnode] (n2) at (1,0) {};
\node[henode] (e1) at (0,-0.5) {$E$};
\node[henode] (e2) at (1,-0.5) {$T$};
\node[henode] (e3) at ($(n1)!0.5!(n2)$) {$OF$};
\draw[undiredge] (n1) -- (e1);
\draw[undiredge] (n2) -- (e2);
\draw[graphedge] (n1) -- (e3) -- (n2);

\end{tikzpicture}%
}
  \end{minipage}}%
\caption{Two other error graphs computed by the backward search}
\label{fig:main-example-another-errors}
\end{minipage}
\end{figure}

\section{Optimizations}
In this section we discuss and formalize some optimizations
that can be applied to the basic backward procedure described in 
the previous section.

\subsubsection*{Lifting the Application Condition to a Post Conditions}
In Procedure~\ref{procedure:back-step} the application condition is checked in 
Step~\ref{backstep:calculate-pocs} for each pushout complement. However, by 
lifting the application condition over the instantiation we can check 
beforehand whether the backward step yields new graphs. We show the lifting in 
the following lemma. 

\newcommand{\lemmaPreToPostCondition}{%
Let $\rho$ be a rule, $(\pi : L \ito \overline{L}, \gamma : \overline{L} \pto 
\overline{R})$ an instantiation of $\rho$ and $m : \overline{R} \ito G$ a 
co-match of the instantiation to some graph $G$. If there is a node $x \in 
\qNodes(\rho)$ where $m(\gamma(\pi(x))$ is defined and attached to an edge $e$ 
without preimage in $\overline{R}$, then there is no pushout complement $H$ of 
$\gamma$, $m$ satisfying the condition of 
Definition~\ref{def:uqrule-application}.}
\begin{lemma}\label{lem:pre-to-post-condition}
\lemmaPreToPostCondition
\end{lemma}

\subsubsection*{Tightening the Upper Bound of Instantiations}
The bound on the length of instantiations proven to exist in 
Proposition~\ref{prop:inst-bound-exists} can be improved depending on the rule 
used. Let $\rho = (r : L \pto R, U)$ be a rule. Obviously $\instBound[\rho]{G} 
= 0$ if $U = \emptyset$. The same holds if instantiations only increase the 
left side of the rule, i.e.~for every $u \in U$ given the instantiation 
$(id_L,r) \instAdd u = (\pi : L \ito \overline{L}_u, \gamma : \overline{L}_u 
\pto \overline{R}_u)$, the graphs $\overline{R}_u$ and $R$ are isomorphic.

A more common situation is that quantifications do not add edges to the right 
side of the instantiations which are solely incident to nodes of the original 
rule $r$. This is case for all rules used in Example~\ref{example:main}. The 
bound can be reduced as shown below.

\newcommand{\lemmaBetterBound}{%
Let $\rho = (r : L \pto R, U)$ and let $(id_L,r) \instAdd u = (\pi : L \ito 
\overline{L}_u, \gamma : \overline{L}_u \pto \overline{R}_u)$. If for every $u 
\in U$ every edge $e \in \overline{R}_u$ without preimage in $R$ is connected 
to a node $v \in \overline{R}_u$ without preimage in $R$, then 
$\instBound[\rho]{G} = |V_G|$.}
\begin{lemma}\label{lem:better-bound}
\lemmaBetterBound
\end{lemma}

\subsubsection*{Optimization by Preparation}
The general framework in \cite{wsts-gts-framework} uses a preparation step in 
the backward search to compute the concatenation of rules and subgraph 
morphisms performed in Step~\ref{backstep:concat-subgraphs} of 
Procedure~\ref{procedure:back-step}. This is not fully possible with 
universally quantified rules since the instantiations are generated within the 
backward steps. However, the preparation step can be performed for rules 
without universal quantifications. For rules with quantification the inner rule 
morphism can be concatenated with subgraph morphisms to partially prepare the 
rule. It can also be show that any concatenation of an instantiation and a 
subgraph morphism which is also a subgraph morphism, will not yield new graph 
in the backward step and thus can be dropped. This also holds for rules with 
universal quantification if all possible instantiations are also subgraph 
morphisms.

\section{Conclusions and Related Work}
In this paper we introduced a categorical formalization for an extension of 
graph transformation systems with universally quantified rules built on the 
single pushout approach. These rules are powerful  enough to model distributed 
algorithms which use broadcast communication. A similar concept are adaptive 
star grammars \cite{dhjme:adaptive-star-grammars} where the left-hand side of a 
rule is a star, i.e.~a designated center node connected to a set of other 
nodes. Arbitrary large graphs can be matched by cloning parts of the star, 
which is -- apart of the restriction to stars -- one of the main differences to 
our approach. Technically our instantiations are a special form of amalgamated 
graph transformations \cite{Boehm1987377}, a technique to merge rules.

The backward search procedure presented in this paper is an extension of 
\cite{wsts-gts-framework} with universally quantified rules and can be used for 
the verification of distributed algorithms, similar to 
\cite{dsz:verification-ad-hoc-networks}. There the induced subgraph ordering 
was used, which was also shown to be compatible with the framework in 
\cite{wsts-gts-framework}. However, our quantifications differ as we have a 
stronger negative application condition such that the induced subgraph ordering 
is not enough to cause our UGTS to satisfy the compatibility condition. This 
also causes the approached to differ in expressiveness. In general our approach 
should be compatible with the induced subgraph ordering and the minor ordering, 
but we did not yet investigated this.

Parameterized verification of combinations of automata- and graph-based models 
of distributed systems has been studied, e.g.~in 
\cite{DDGLZ09,ADR11,DSZ10,DSZ11,DSTZ12,DST13}. In \cite{AHDR08} we applied graph-based 
transformations to model intermediate evaluations of non-atomic mutual 
exclusion protocols with universally quantified conditions. The conditions are 
not defined however in terms of graph rewrite rules. Semi-decision procedures 
can be defined by resorting to upward closed abstractions during backward 
search (monotonic abstraction as in \cite{DR12}). In \cite{DDGLZ09} we 
studied decidability of reachability and coverability for a graph-based 
specification used to model biological systems. Among other results, we proved 
undecidability for coverability for graph rewrite systems that can only 
increase the size of a configuration. Reachability problems for graph-based 
representations of protocols have also been considered in \cite{ADR11} where 
symbolic representations combining a special graph ordering and 
constraint-based representation of relations between local data of different 
nodes have been used to verify parameterized consistency protocols. 
Coverability for GTS is studied in \cite{RTA12} where it was proved that it is 
decidable for bounded path graphs ordered via subgraph inclusion. A model with 
topologies represented as acyclic directed graphs has been presented in 
\cite{AAR13}. Coverability for automata-based models of broadcast communication 
has recently been studied in \cite{DSZ10,DSZ11,DSTZ12,DT13,DST13}. 
In the context of program analysis approximated backward search working on graphs 
representing data structures with pointers have been considered in \cite{ACV13}.
In this setting approximations are defined via edges or node deletion.

\bibliography{broadcast-subgraph}

\begin{thebibliography}{10}

\bibitem{AAR13}
P.~A. Abdulla, M.~F. Atig, and O.~Rezine.
\newblock Verification of directed acyclic ad hoc networks.
\newblock In {\em FMOODS/FORTE}, pages 193--208, 2013.

\bibitem{ACV13}
P.~A. Abdulla, J.~Cederberg, and T.~Vojnar.
\newblock Monotonic abstraction for programs with multiply-linked structures.
\newblock {\em Int. J. Found. Comput. Sci.}, 24(2):187--210, 2013.

\bibitem{ADR09}
P.~A. Abdulla, G.~Delzanno, and A.~Rezine.
\newblock Approximated parameterized verification of infinite-state processes
  with global conditions.
\newblock {\em Formal Methods in System Design}, 34(2):126--156, 2009.

\bibitem{ADR11}
P.~A. Abdulla, G.~Delzanno, and A.~Rezine.
\newblock Automatic verification of directory-based consistency protocols with
  graph constraints.
\newblock {\em Int. J. Found. Comput. Sci.}, 22(4), 2011.

\bibitem{AHDR08}
P.~A. Abdulla, N.~Ben Henda, G.~Delzanno, and A.~Rezine.
\newblock Handling parameterized systems with non-atomic global conditions.
\newblock In {\em VMCAI'08}, volume 4905 of {\em LNCS}, pages 22--36. Springer,
  2008.

\bibitem{acjt:general-decidability}
P.~A. Abdulla, K.~\u{C}er\={a}ns, B.~Jonsson, and Y.~Tsay.
\newblock General decidability theorems for infinite-state systems.
\newblock In {\em Proc. of LICS '96}, pages 313--321. IEEE, 1996.

\bibitem{BDKSS12}
N.~Bertrand, G.~Delzanno, B.~K{\"o}nig, A.~Sangnier, and J.~St{\"u}ckrath.
\newblock On the decidability status of reachability and coverability in graph
  transformation systems.
\newblock In {\em RTA'12}, volume~15 of {\em LIPIcs}, pages 101--116. Schloss
  Dagstuhl - Leibniz-Zentrum fuer Informatik, 2012.

\bibitem{RTA12}
N.~Bertrand, G.~Delzanno, B.~K{\"o}nig, A.~Sangnier, and J.~St{\"u}ckrath.
\newblock On the decidability status of reachability and coverability in graph
  transformation systems.
\newblock In {\em RTA}, pages 101--116, 2012.

\bibitem{Boehm1987377}
P.~Boehm, H.~Fonio, and A.~Habel.
\newblock Amalgamation of graph transformations: A synchronization mechanism.
\newblock {\em Journal of Computer and System Sciences}, 34:377 -- 408, 1987.

\bibitem{DDGLZ09}
G.~Delzanno, C.~Di Giusto, M.~Gabbrielli, C.~Laneve, and G.~Zavattaro.
\newblock The {\it kappa}-lattice: Decidability boundaries for qualitative
  analysis in biological languages.
\newblock In {\em CMSB}, pages 158--172, 2009.

\bibitem{DR12}
G.~Delzanno and A.~Rezine.
\newblock A lightweight regular model checking approach for parameterized
  systems.
\newblock {\em STTT}, 14(2):207--222, 2012.

\bibitem{DST13}
G.~Delzanno, A.~Sangnier, and R.~Traverso.
\newblock Parameterized verification of broadcast networks of register
  automata.
\newblock In {\em RP'13}, pages 109--121, 2013.

\bibitem{DSTZ12}
G.~Delzanno, A.~Sangnier, R.~Traverso, and G.~Zavattaro.
\newblock On the complexity of parameterized reachability in reconfigurable
  broadcast networks.
\newblock In {\em FSTTCS'12}, volume~18 of {\em LIPIcs}, pages 289--300.
  Schloss Dagstuhl - Leibniz-Zentrum fuer Informatik, 2012.

\bibitem{dsz:verification-ad-hoc-networks}
G.~Delzanno, A.~Sangnier, and G.~Zavattaro.
\newblock Parameterized verification of ad hoc networks.
\newblock In {\em Proc. CONCUR '10}, pages 313--327. Springer, 2010.
\newblock {LNCS} 6269.

\bibitem{DSZ10}
G.~Delzanno, A.~Sangnier, and G.~Zavattaro.
\newblock Parameterized verification of ad hoc networks.
\newblock In {\em CONCUR'10}, volume 6269 of {\em LNCS}, pages 313--327.
  Springer, 2010.

\bibitem{DSZ11}
G.~Delzanno, A.~Sangnier, and G.~Zavattaro.
\newblock On the power of cliques in the parameterized verification of ad hoc
  networks.
\newblock In {\em FOSSACS'11}, volume 6604 of {\em LNCS}, pages 441--455.
  Springer, 2011.

\bibitem{DS2014}
G.~Delzanno and J.~St{\"u}ckrath.
\newblock Parameterized verification of graph transformation systems with whole
  neighbourhood operations.
\newblock In {\em RP'14}, 2014.

\bibitem{DT13}
G.~Delzanno and R.~Traverso.
\newblock Decidability and complexity results for verification of asynchronous
  broadcast networks.
\newblock In {\em LATA}, pages 238--249, 2013.

\bibitem{d:subgraphs-wqo}
G.~Ding.
\newblock Subgraphs and well-quasi-ordering.
\newblock {\em Jornal of Graph Theory}, 16:489--502, November 1992.

\bibitem{dhjme:adaptive-star-grammars}
F.~Drewes, B.~Hoffmann, D.~Janssens, M.~Minas, and N.~V. Eetvelde.
\newblock Adaptive star grammars.
\newblock In {\em Proc. of ICGT '06 (International Conference on Graph
  Transformation)}, pages 77--91. Springer, 2006.
\newblock LNCS 4178.

\bibitem{handbook-spo}
H.~Ehrig, R.~Heckel, M.~Korff, M.~L\"{o}we, L.~Ribeiro, A.~Wagner, and
  A.~Corradini.
\newblock Algebraic approaches to graph transformation---part~{II}: Single
  pushout approach and comparison with double pushout approach.
\newblock In G.~Rozenberg, editor, {\em Handbook of Graph Grammars and
  Computing by Graph Transformation, Vol.1: Foundations}, chapter~4. World
  Scientific, 1997.

\bibitem{fs:well-structured-everywhere}
A.~Finkel and P.~Schnoebelen.
\newblock Well-structured transition systems everywhere!
\newblock {\em Theoretical Computer Science}, 256(1-2):63--92, April 2001.

\bibitem{HJKS:pocs2010}
M.~Heum{\"u}ller, S.~Joshi, B.~K{\"o}nig, and J.~St{\"u}ckrath.
\newblock Construction of pushout complements in the category of hypergraphs.
\newblock In {\em Proc. of GCM '10 (Workshop on Graph Computation Models)},
  2010.

\bibitem{jk:minor-wqo}
S.~Joshi and B.~K{\"o}nig.
\newblock Applying the graph minor theorem to the verification of graph
  transformation systems.
\newblock In {\em Proc. of CAV '08}, pages 214--226. Springer, 2008.
\newblock {LNCS} 5123.

\bibitem{wsts-gts-framework}
Barbara K{\"o}nig and Jan St{\"u}ckrath.
\newblock A general framework for well-structured graph transformation systems.
\newblock In P.~Baldan and D.~Gorla, editors, {\em Proc. of CONCUR 2014},
  volume 8704 of {\em LNCS}, pages 467--481. Springer, 2014.

\bibitem{NamjoshiTrefler}
K.~S. Namjoshi and R.~J. Trefler.
\newblock Uncovering symmetries in irregular process networks.
\newblock In Roberto Giacobazzi, Josh Berdine, and Isabella Mastroeni, editors,
  {\em VMCAI}, volume 7737 of {\em Lecture Notes in Computer Science}, pages
  496--514. Springer, 2013.

\end{thebibliography}
\bibliographystyle{plain}

\opt{long}{

\newpage
\appendix

\section{Important Properties of Pushouts}
Pushouts and pushout complements are a well-known notion in category theory and 
are the basis for the single pushout approach as well as the double pushout 
approach \cite{handbook-spo}. We briefly recall the definition of a pushout and 
a few important properties of such.

\noindent
\parbox{0.7\textwidth}{
\begin{definition}\label{def:pushout}
  Let $\phi\colon G_0\pto G_1$ and $\psi\colon G_0\pto G_2$ be two
  partial graph morphisms. The \emph{pushout} of $\phi$ and $\psi$
  consists of a graph $G_3$ and two morphisms $\psi'\colon G_1\pto
  G_3$, $\phi'\colon G_2\pto G_3$ such that $\psi'\circ\phi =
  \phi'\circ\psi$ and for every other pair of morphisms $\psi''\colon
  G_1\pto G'_3$, $\phi''\colon G_2\pto G'_3$ such that
  $\psi''\circ\phi = \phi''\circ\psi$ there exists a unique morphism
  $\eta\colon G_3\pto G'_3$ with $\eta\circ\psi' = \psi''$ and
  $\eta\circ\phi' = \phi''$.
\end{definition}}
\parbox{0.3\textwidth}{
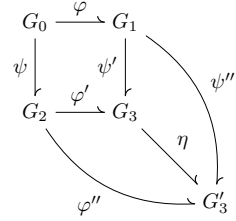
\begin{center}
  \scalebox{0.8}{\begin{tikzpicture}[x=1.5cm,y=-1.5cm]

\node (G0) at (0,0) {$G_0$};
\node (G1) at (1,0) {$G_1$};
\node (G2) at (0,1) {$G_2$};
\node (G3) at (1,1) {$G_3$};
\node (G3b) at (2,2) {$G_3'$};
\draw[mor-parl] (G0) -- node [midway, above] {$\phi$} (G1);
\draw[mor-parl] (G0) -- node [midway, left] {$\psi$} (G2);
\draw[mor-parl] (G2) -- node [midway, above] {$\phi'$} (G3);
\draw[mor-parl] (G1) -- node [midway, left] {$\psi'$} (G3);
\draw[mor-parl] (G3) -- node [midway, above right] {$\eta$} (G3b);
\draw[mor-parl] (G1) to[bend left] node [midway, above right] {$\psi''$} (G3b);
\draw[mor-parl] (G2) to[bend right] node [midway, below left] {$\phi''$} (G3b);

\end{tikzpicture}%
}
\end{center}}


Without proof we use the following properties of pushouts in our proofs in 
Appendix~\ref{app:proofs}.

\begin{lemma}\label{lem:preserve-injection}
Let $\phi : G_1 \pto G_2$ and $\psi : G_1 \ito G_3$ be morphisms and let $\phi' 
: G_3 \pto G_4$, $\psi' : G_2 \ito G_4$ be their pushout. If $\psi$ is total 
and injective, then $\psi'$ is also total and injective.
\end{lemma}

\begin{lemma}\label{lem:pushout-stacking}
Let morphisms as shown in the diagram below be given. It can be shown that the 
following two properties hold for any category.
\begin{enumerate}
  \item If the left and the right squares below are pushouts, the outer square 
  is a pushout as well.
  \item If the left square and the outer square are pushouts, the right square 
  is a pushout as well.
\end{enumerate}
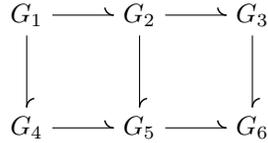
\begin{center}
  \scalebox{1.0}{\begin{tikzpicture}[x=1.5cm,y=-1.5cm]

\node (G1) at (0,0) {$G_1$};
\node (G2) at (1,0) {$G_2$};
\node (G3) at (2,0) {$G_3$};
\node (G4) at (0,1) {$G_4$};
\node (G5) at (1,1) {$G_5$};
\node (G6) at (2,1) {$G_6$};
\draw[mor-parl] (G1) -- (G2);
\draw[mor-parl] (G2) -- (G3);
\draw[mor-parl] (G1) -- (G4);
\draw[mor-parl] (G2) -- (G5);
\draw[mor-parl] (G3) -- (G6);
\draw[mor-parl] (G4) -- (G5);
\draw[mor-parl] (G5) -- (G6);

\end{tikzpicture}%
}
\end{center}
\end{lemma}

\section{Proofs}\label{app:proofs}

\begin{lemma}\label{lem:rule-instantiation-swap}
Let $\rho = (r,U)$ be a rule and let $f : U \to \nat_0$ be any function 
assigning a quantity to each universal quantification. Every instantiation 
of $\rho$ which is generated by using $f(u)$ occurrences for each $u$ 
respectively, yields the same morphisms (up to isomorphism).
\end{lemma}

\begin{proof}
We show this property by showing that we can swap each two instantiation steps 
without changing the instantiation containing both steps.
Let $\iota = (\pi, \gamma)$ be an instantiation of some rule $\rho = (r, 
U)$ and let $u = (p_u, q_u),  v = (p_v, q_v) \in U$ be two universal 
quantifications as shown in the diagram below. There the upper part of the 
diagram is the instantiation $\iota_v = \iota \instAdd v$, while the front 
part of the diagram is the instantiation $\iota_u = \iota \instAdd u$.

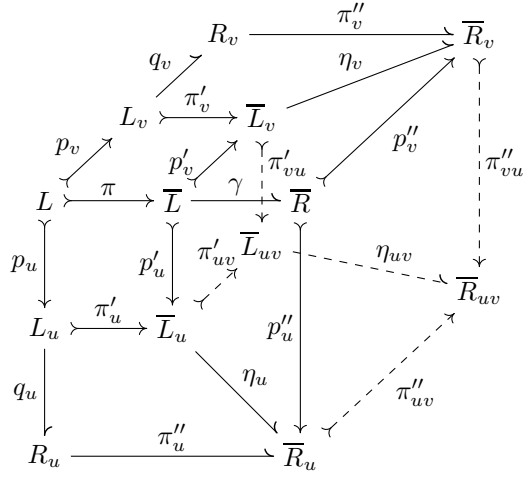
\gcom{It seems to me that the part of the diagram needed in the proof is much smaller
than the diagram displayed below.
We basically use the "western" inner diagrams pointing to $\overline{L}_{uv}$
and the "western" outer diagrams pointing to $\overline{R}_{uv}$
We could drop or write in grey the remaining parts to highlight the morphisms
on which we want to show permutation
Apart from this...it is nice to show permutation using commutation!
The proof is ok to me.}
\begin{center}
  \scalebox{1.0}{\begin{tikzpicture}[x=1.7cm,y=-1.7cm]

\begin{scope}
  \node (L) at (0,0) {$L$};
  \node (oL) at (1,0) {$\overline{L}$};
  \node (oR) at (2,0) {$\overline{R}$};
  \node (Li) at (0,1) {$L_u$};
  \node (oLi) at (1,1) {$\overline{L}_u$};
  \node (Ri) at (0,2) {$R_u$};
  \node (oRi) at (2,2) {$\overline{R}_u$};
  
  \draw[mor-tot-inj] (L) -- node [morlabel, above] {$\pi$} (oL);
  \draw[mor-parl] (oL) -- node [morlabel, above] {$\gamma$} (oR);
  \draw[mor-tot-inj] (L) -- node [morlabel, left] {$p_u$} (Li);
  \draw[mor-tot-inj] (oL) -- node [morlabel, left] {$p_u'$} (oLi);
  \draw[mor-tot-inj] (Li) -- node [morlabel, above] {$\pi_u'$} (oLi);
  \draw[mor-parl] (Li) -- node [morlabel, left] {$q_u$} (Ri);
  \draw[mor-parl] (oLi) -- node [morlabel, above right] {$\eta_u$} (oRi);
  \draw[mor-tot-inj] (oR) -- node [morlabel, left] {$p_u''$} (oRi);
  \draw[mor-parl] (Ri) -- node [morlabel, above] {$\pi_u''$} (oRi);
\end{scope}

\begin{scope}[shift={(0.7,-0.65)}]
  \node (Lj) at (0,0) {$L_v$};
  \node (oLj) at (1,0) {$\overline{L}_v$};
  \node (oLij) at (1,1) {$\overline{L}_{uv}$};
  \draw[mor-tot-inj] (L) -- node [morlabel, above left] {$p_v$} (Lj);
  \draw[mor-tot-inj] (oL) -- node [morlabel, left, inner sep=7pt] {$p_v'$} 
  (oLj);
  \draw[mor-tot-inj, dashed] (oLi) -- node [morlabel, above, inner sep=8pt] 
  {$\pi_{uv}'$} (oLij);
  \draw[mor-tot-inj] (Lj) -- node [morlabel, above] {$\pi_v'$} (oLj);
  \draw[mor-tot-inj, dashed] (oLj) -- node [morlabel, above right] 
  {$\pi_{vu}'$} (oLij);
\end{scope}

\begin{scope}[shift={(1.4,-1.3)}]
  \node (Rj) at (0,0) {$R_v$};
  \node (oRj) at (2,0) {$\overline{R}_v$};
  \node (oRij) at (2,2) {$\overline{R}_{uv}$};
  \draw[mor-parl] (Lj) -- node [morlabel, above left] {$q_v$} (Rj);
  \draw[mor-parl] (oLj) -- node [morlabel, above left] {$\eta_v$} (oRj);
  \draw[mor-tot-inj] (oR) -- node [morlabel, below right] {$p_v''$} (oRj);
  \draw[mor-parl, dashed] (oLij) -- node [morlabel, above right] {$\eta_{uv}$} 
  (oRij);
  \draw[mor-tot-inj, dashed] (oRi) -- node [morlabel, below right] 
  {$\pi_{uv}''$} (oRij);
  \draw[mor-parl] (Rj) -- node [morlabel, above] {$\pi_v''$} (oRj);
  \draw[mor-tot-inj, dashed] (oRj) -- node [morlabel, right] {$\pi_{vu}''$} 
  (oRij);
\end{scope}

\end{tikzpicture}%
}
\end{center}

Let $\overline{L}_{uv}$ be the pushout of $p_u'$, $p_v'$ and let 
$\overline{R}_{uv}$ be the pushout of $p_u''$, $p_v''$. By the properties 
of pushouts a unique $\eta_{uv}$ exists and we will show that $\iota_u 
\instAdd 
v = (\pi_{vu}' \circ \pi_v' \circ p_v, \eta_{uv}) = \iota_v \instAdd u$.

By construction all squares $p_u' \circ \pi = \pi_u' \circ p_u$, 
$p_v' \circ \pi = \pi_v' \circ p_v$ and $\pi_{vu}' \circ p_v' = \pi_{uv}' 
\circ 
p_u'$ are pushouts. Therefore, the squares $\pi_{vu}' \circ \pi_v' \circ p_v = 
\pi_{uv}' \circ p_u' \circ \pi$ and $\pi_{vu}' \circ p_v' \circ \pi = 
\pi_{uv}' 
\circ \pi_u' \circ p_u$ are pushouts as well. Thus, $\overline{L}_{uv}$ is the 
pushout of $p_v$, $p_u' \circ \pi$ computed in the construction of $\iota_u 
\instAdd v$ as well as the pushout of $p_u$, $p_v' \circ \pi$ computed in the 
construction of $\iota_v \instAdd u$.
The same property holds for $\overline{R}_{uv}$ using the three large outer 
squares. Since $\eta_{uv}$ is unique, both sequences of the instantiation 
steps 
give rise to the same morphisms.
This means that every instantiation can be uniquely characterized only by the 
number on instantiation steps for each $u \in U$. \qed
\end{proof}

\begin{lemma}\label{lem:large2small-rule-subgraph-existence}
\extranote{If $q_u$ can map quantified nodes to non-quantified nodes, $\mu_u'$ 
and $\mu_u''$ still exist, but do not commute.}
Let $\rho = (r,U)$ be a rule and let $(\pi : L \to \overline{L}, \gamma : 
\overline{L} \pto \overline{R})$ be an instantiation of $\rho$. For every 
further instantiation $(\pi, \gamma) \instAdd u$ using some $u \in U$, there 
are two subgraph morphisms $\mu_u' : \overline{L}_u \subArrow \overline{L}$ 
and 
$\mu_u'' : \overline{R}_u \subArrow \overline{R}$ such that $\gamma \circ 
\mu_u' = \mu_u'' \circ \eta$.
\end{lemma}

\begin{proof}
\gcom{
This proofs seems too complicate for what is proved.\\
Shall we try to simplify it?
}
By definition $p_u$ and $q_u \circ p_u$ are total and injective, 
thus, $p_u'$ and $q_u'' \circ p_u''$ are total and injective as 
well (see Lemma~\ref{lem:preserve-injection}). Hence, the reverse morphisms 
$\mu_u'$ and $\mu_u''$ are partial, injective and surjective, i.e.~subgraph 
morphisms. By using Lemma~\ref{lem:pushout-stacking} it can be shown that by 
forming the pushout $\overline{R}'$ of $p_u'$ and $\gamma$, the pushout 
$R_u'$ of $q_u$ and $\pi'$ and then the pushout $\overline{R}_u$ of $q_u'$ 
and $\gamma'$, we obtain the same graph $\overline{R}_u$ as by forming the 
pushout of $q_u \circ p_u$ and $\gamma \circ \pi$ directly. Furthermore, 
the diagram below commutes with the exception of $\mu_u'$ and $\mu_u''$, for 
which we still have to show the commutativity with $\eta$ and $\gamma$.
\begin{center}
  \scalebox{1.0}{\begin{tikzpicture}[x=1.5cm,y=-1.5cm]

\node (L) at (0,0) {$L$};
\node (oL) at (1,0) {$\overline{L}$};
\node (oR) at (2,0) {$\overline{R}$};
\node (Li) at (0,1) {$L_u$};
\node (oLi) at (1,1) {$\overline{L}_u$};
\node (Ri2) at (1,2) {$R_u'$};
\node (Ri) at (0,2) {$R_u$};
\node (oR2) at (2,1) {$\overline{R}'$};
\node (oRi) at (2,2) {$\overline{R}_u$};

\draw[mor-tot-inj] (L) -- node [morlabel, above] {$\pi$} (oL);
\draw[mor-parl] (oL) -- node [morlabel, above] {$\gamma$} (oR);
\draw[mor-tot-inj] (L) -- node [morlabel, left] {$p_u$} (Li);
\draw[mor-tot-inj] (oL) -- node [morlabel, left] {$p_u'$} (oLi);
\draw[mor-tot-inj] (Li) -- node [morlabel, above] {$\pi'$} (oLi);
\draw[mor-parl] (Li) -- node [morlabel, left] {$q_u$} (Ri);
\draw[mor-parl] (oLi) -- node [morlabel, above right] {$\eta$} (oRi);
\draw[mor-tot-inj] (oR) -- node [morlabel, right] {$p_u''$} (oR2);
\draw[mor-tot-inj] (Ri) -- node [morlabel, above] {$\pi''$} (Ri2);
\draw[mor-parl] (oLi) -- node [morlabel, above] {$\gamma'$} (oR2);
\draw[mor-parl] (oLi) -- node [morlabel, left] {$q_u'$} (Ri2);
\draw[mor-parl] (Ri2) -- node [morlabel, above] {$\gamma''$} (oRi);
\draw[mor-parl] (oR2) -- node [morlabel, right] {$q_u''$} (oRi);
\draw[mor-subgraph] (oLi) to[bend right] node [morlabel, right] {$\mu_u'$} (oL);
\draw[mor-subgraph] (oRi) to[bend right=50] node [morlabel, right] {$\mu_u''$} 
(oR);

\end{tikzpicture}%
}
\end{center}

Let $x \in \overline{L}_u$ and assume $\gamma(\mu_u'(x))$ is defined. This 
means that there is exactly one $x' \in \overline{L}$ with $p_u'(x') = x$ and 
$\gamma(x')$ is defined. Since $q_u'' \circ p_u''$ is total and injective, 
there is an $x'' \in \overline{R}_u$ with $q_u''(p_u''(\gamma(x'))) = x''$. 
Due to commutativity of the diagram we obtain $\eta(x) = x''$. Hence, we know 
that $\mu_u''(\eta(x))$ is defined and $\gamma(\mu_u'(x)) = \mu_u''(\eta(x))$.

Now assume $\mu_u'(x)$ is defined, but $\gamma(\mu_u'(x))$ is undefined. 
Because of commutativity, $\gamma'(x)$ is undefined as well and therefore also 
$\eta(x) = q_u''(\gamma'(x))$ is undefined. Thus, $\gamma(\mu_u'(x)) = 
\mu_u''(\eta(x))$ are both undefined.

Now assume $\mu_u'(x)$ is undefined. If $\eta(x)$ is undefined, 
$\gamma(\mu_u'(x)) = \mu_u''(\eta(x))$ are obviously both undefined, so we 
assume that $\eta(x)$ is defined and show that $\eta(x)$ has no preimages under 
$q_u'' \circ p_u''$. For this we only have to consider elements in 
$\overline{R}'$ which have preimages in $\overline{L}_u$, since an element 
without a preimage and mapped to $\eta(x)$ would violate the pushout property 
of the lower right square.
We observe that $\gamma'(x)$ has no preimage in $\overline{R}$, since the top 
right square would not be a pushout. In fact this holds for every $x' \in 
\overline{L}_u$ with $\eta(x') = \eta(x)$ if $x'$ has no preimage in 
$\overline{L}$. By the same argument we also know that $\gamma'(x')$ has 
exactly one preimage in $\overline{L}_u$. This means that two $x'$ with and 
without preimage in $\overline{L}$ are not merged by $\gamma'$. By showing that 
these $x'$ are also not merged by $q_u'$, we know that their image in the 
pushout $\overline{R}_u$ would not be equal and prove that there are in fact no 
$x'$ with preimage in $\overline{L}$.

If $x'$ has a preimage in $\overline{L}$ but not in $L_u$, then $x'$ is not 
merged with any other element by $q_u'$, since the left lower square is a 
pushout. If $x'$ has a preimage in $\overline{L}$ and $L_u$, it also has 
(exactly) one preimage in $L$, because of the top left square being a pushout. 
Thus, by Definition~\ref{def:uqrule} $q_u$ may not merge the preimage of $x'$ 
with anything else, especially not with the preimage of $x$. Since neither 
$\pi'$ nor $q_u$ merge the preimage of $x'$ with anything, $x$ is not merged 
with anything via $q_u'$ as well. Thus, $\eta(x') = \eta(x)$ cannot hold and 
$\eta(x)$ has no preimage in $\overline{R}$.

We have shown that $\gamma(\mu_u'(x))$ is undefined if and only if 
$\mu_u''(\eta(x))$ is undefined, thus the commutativity $\gamma \circ \mu_u' = 
\mu_u'' \circ \eta$ follows from $\mu_u'$ being the reverse of $p_u'$ and 
$\mu_u''$ being the reverse of $q_u'' \circ p_u''$. \qed
\end{proof}

\begin{lemma}\label{lem:subgraph-split-rule-minimization}
Let $\rho = (r,U)$ be a rule and let $(\pi_i : L \ito L_i, \gamma_i : L_i 
\pto R_i)$ for $i \in \{1,2\}$ be two instantiations of $\rho$ with 
$(\pi_2, \gamma_2) = (\pi_1, \gamma_1) \instAdd u$ for some $u \in U$. 
Furthermore, let $\mu_L : L_2 \subArrow L_1$, $\mu_R : R_2 \subArrow R_1$, 
$\mu_R' : R_1 \subArrow R$ be subgraph morphisms with $\gamma_1 \circ \mu_L = 
\mu_R \circ \gamma_2$ and let $m : R \ito G$ be a match.
For every pushout complement $H_2$ of $\mu_R' \circ \mu_R \circ \gamma_2$ and 
$m$ where $m_2' : L_2 \ito H_2$ is total and injective, there is a pushout 
complement $H_1$ of $\mu_R' \circ \gamma_1$ and $m$ with $H_1 \subOrder H_2$.
\end{lemma}

\begin{proof}
We will show this by using the fact, that subgraph morphisms are preserved by 
total pushouts and successively building the commuting diagram below.
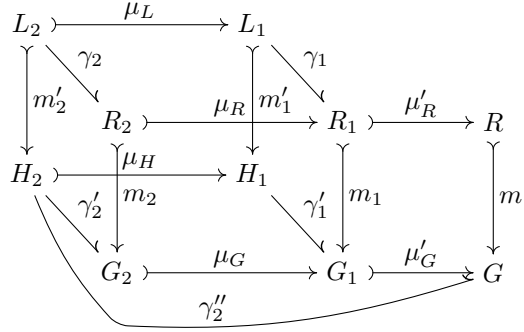
\begin{center}
  \scalebox{1.0}{\begin{tikzpicture}[x=2cm,y=-2cm]

\begin{scope}
  \node (Lnp1) at (0,0) {$L_2$};
  \node (Rnp1) at (0.6,0.65) {$R_2$};
  \node (Hnp1) at (0,1) {$H_2$};
  \node (Gnp1) at (0.6,1.65) {$G_2$};
  \draw[mor-parl] (Lnp1) -- node [morlabel, above right] {$\gamma_2$} (Rnp1);
  \draw[mor-tot-inj] (Lnp1) -- node [morlabel, right] {$m'_2$} (Hnp1);
  \draw[mor-tot-inj] (Rnp1) -- node [morlabel, right] {$m_2$} (Gnp1);
  \draw[mor-parl] (Hnp1) -- node [morlabel, above right] {$\gamma'_2$} (Gnp1);
\end{scope}

\begin{scope}[shift={(1.5,0)}]
  \node (Ln) at (0,0) {$L_1$};
  \node (Rn) at (0.6,0.65) {$R_1$};
  \node (Hn) at (0,1) {$H_1$};
  \node (Gn) at (0.6,1.65) {$G_1$};
  \draw[mor-parl] (Ln) -- node [morlabel, above right] {$\gamma_1$} (Rn);
  \draw[mor-tot-inj] (Ln) -- node [morlabel, right] {$m'_1$} (Hn);
  \draw[mor-tot-inj] (Rn) -- node [morlabel, right] {$m_1$} (Gn);
  \draw[mor-parl] (Hn) -- node [morlabel, above right] {$\gamma'_1$} (Gn);
\end{scope}

\draw[mor-subgraph] (Lnp1) -- node [morlabel, above] {$\mu_L$} (Ln);
\draw[mor-subgraph] (Rnp1) -- node [morlabel, above] {$\mu_R$} (Rn);
\draw[mor-subgraph] (Gnp1) -- node [morlabel, above] {$\mu_G$} (Gn);
\draw[mor-subgraph] (Hnp1) -- node [morlabel, above] {$\mu_H$} (Hn);

\begin{scope}[shift={(3.1,0.65)}]
  \node (R) at (0,0) {$R$};
  \node (G) at (0,1) {$G$};
  \draw[mor-tot-inj] (R) -- node [morlabel, right] {$m$} (G);
\end{scope}

\draw[mor-subgraph] (Rn) -- node [morlabel, above] {$\mu_R'$} (R);
\draw[mor-subgraph] (Gn) -- node [morlabel, above] {$\mu_G'$} (G);
\draw[mor-parl, rounded corners] (Hnp1) to[bend right=10] ($(Gnp1)+(0,0.35)$)  
to[bend right=10] node [morlabel, near start, above] {$\gamma_2''$} (G);

\end{tikzpicture}%
}
\end{center}

Let $H_2$ be a pushout complement of $\mu_R' \circ \mu_R \circ \gamma_2$ and 
$m$, where $m_2'$ is total and injective. We compute the pushout $G_2$ of 
$\gamma_2$ and $m_2'$ and then the pushout $G_1$ of $\mu_R$ and $m_2$. Due to 
Lemma~\ref{lem:pushout-stacking}, $G_1$ is also the pushout of $\mu_R \circ 
\gamma_2$ and $m_2'$ and therefore there is a unique $\mu_G' : G_1 \subArrow G$ 
such that the diagram commutes. Since $G$ is the pushout of $\mu_R' \circ \mu_R 
\circ \gamma_2$ and $m_2'$, the rightmost square is in fact a pushout as well.
Now form the pushout $H_1$ of $\mu_L$ and $m_2'$. Again the existence of 
$\gamma_1'$ follows from the pushout properties. Since the diagram $m_1 \circ 
\gamma_1 \circ \mu_L = \gamma_1' \circ \mu_H \circ m_2'$ commutes with the 
pushout $m_1 \circ \mu_R \circ \gamma_2 = \mu_G \circ \gamma_2' \circ m_2'$, it 
is also a pushout and hence, $G_1$ is a pushout of $\gamma_1$ and $m_1'$. This 
means that $H_1$ is in fact a pushout complement of $\mu_R' \circ \gamma_1$ and 
$m$. Since subgraph morphisms are preserved by total pushouts, $\mu_H$ is a 
subgraph morphism. Thus, $H_1 \subOrder H_2$. \qed
\end{proof}

\begin{proposition_app}{prop:inst-bound-exists}
\propInstBoundExists
\end{proposition_app}

\begin{proof}
Let $\iota_{k-1} = (\pi : L \ito \overline{L}, \gamma : \overline{L} \pto 
\overline{R})$ be a rule instantiation of length $k-1$ of $(r,U)$ such that 
$\iota_k = (\pi, \gamma) \instAdd u$ for some $u \in U$, let $\nu : 
\overline{R}_u \subArrow R$ be a subgraph morphism and let $m : R \ito G$ be a 
co-match as shown in the diagram below.
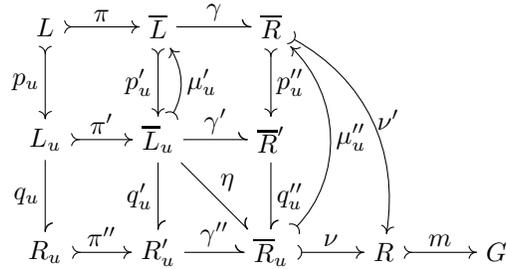
\begin{center}
  \scalebox{1.0}{\begin{tikzpicture}[x=1.5cm,y=-1.5cm]

\node (L) at (0,0) {$L$};
\node (oL) at (1,0) {$\overline{L}$};
\node (oR) at (2,0) {$\overline{R}$};
\node (Li) at (0,1) {$L_u$};
\node (oLi) at (1,1) {$\overline{L}_u$};
\node (Ri2) at (1,2) {$R_u'$};
\node (Ri) at (0,2) {$R_u$};
\node (oR2) at (2,1) {$\overline{R}'$};
\node (oRi) at (2,2) {$\overline{R}_u$};
\node (R) at (3,2) {$R$};
\node (G) at (4,2) {$G$};

\draw[mor-tot-inj] (L) -- node [morlabel, above] {$\pi$} (oL);
\draw[mor-parl] (oL) -- node [morlabel, above] {$\gamma$} (oR);
\draw[mor-tot-inj] (L) -- node [morlabel, left] {$p_u$} (Li);
\draw[mor-tot-inj] (oL) -- node [morlabel, left] {$p_u'$} (oLi);
\draw[mor-tot-inj] (Li) -- node [morlabel, above] {$\pi'$} (oLi);
\draw[mor-parl] (Li) -- node [morlabel, left] {$q_u$} (Ri);
\draw[mor-parl] (oLi) -- node [morlabel, above right] {$\eta$} (oRi);
\draw[mor-tot-inj] (oR) -- node [morlabel, right] {$p_u''$} (oR2);
\draw[mor-tot-inj] (Ri) -- node [morlabel, above] {$\pi''$} (Ri2);
\draw[mor-parl] (oLi) -- node [morlabel, above] {$\gamma'$} (oR2);
\draw[mor-parl] (oLi) -- node [morlabel, left] {$q_u'$} (Ri2);
\draw[mor-parl] (Ri2) -- node [morlabel, above] {$\gamma''$} (oRi);
\draw[mor-parl] (oR2) -- node [morlabel, right] {$q_u''$} (oRi);
\draw[mor-subgraph] (oLi) to[bend right] node [morlabel, right] {$\mu_u'$} (oL);
\draw[mor-subgraph] (oRi) to[bend right=45] node [morlabel, right] {$\mu_u''$} 
(oR);
\draw[mor-subgraph] (oRi) -- node [morlabel, above] {$\nu$} (R);
\draw[mor-subgraph] (oR) to[bend left=35] node [morlabel, right] {$\nu'$} (R);
\draw[mor-tot-inj] (R) -- node [morlabel, above] {$m$} (G);

\end{tikzpicture}%
}
\end{center}

From Lemma~\ref{lem:large2small-rule-subgraph-existence} we know that $\mu_u'$ 
and $\mu_u''$ exist and the diagram commutes. We will show the existence of a 
subgraph morphism $\nu' : \overline{R} \subArrow R$ satisfying $\nu = \nu' 
\circ \mu_u''$. Then from Lemma~\ref{lem:subgraph-split-rule-minimization} it 
follows that every graph computed by a backward step of $\nu \circ \eta$, the 
instantiation $\iota_k$, is already represented by a backward step of 
$\nu' \circ \gamma$, the instantiation $\iota_{k-1}$.

First assume that $\gamma''(\pi''(x_u))$ is undefined for every $x_u \in R_u$ 
which has no preimage under $q_u \circ p_u$. We can show that $q_u'' \circ 
p_u''$ is a subgraph morphism by showing that it is surjective. Assume there 
is an $\overline{x}_u \in \overline{R}_u$ without preimage under $q_u'' \circ 
p_u''$. Since the large square is a pushout, there is an $x_u' \in R_u$ with 
$\gamma''(\pi''(x_u')) = \overline{x}_u$. By the first assumption $x_u'$ must 
have a preimage $x \in L$ under $q_u \circ p_u$ for $\gamma''(\pi''(x_u'))$ to 
be defined. Due to the commutativity, $\gamma(\pi(x))$ is defined and there is 
a preimage of $\overline{x}_u$ in $\overline{R}$, violating the 
second assumption. Hence, $q_u'' \circ p_u''$ is a subgraph morphism 
commuting with $\mu_u''$ (in fact $\overline{R}$ and $\overline{R}_u$ are 
isomorphic). The morphism $\nu' = \nu \circ q_u'' \circ p_u''$ satisfies the 
necessary properties.

If at least one quantification within $\iota_k$ satisfies the previous 
restriction, by Lemma~\ref{lem:rule-instantiation-swap} we can assume 
w.l.o.g.~that it is the last quantification step. So assume for every 
quantification step there is at least one $x_u \in R_u$ without preimage under 
$q_u \circ p_u$ such that $\gamma''(\pi''(x_u))$ is defined. Since 
$\gamma''(\pi''(x_u))$ has no preimage under $q_u'' \circ p_u''$ (otherwise it 
would have a preimage in $L$), the graph $\overline{R}_u$ has at least $k$ 
nodes and edges. Thus, since $R$ has less than $k$ nodes and edges, for at 
least one instantiation step within $\iota_k$ for every $x_u' \in R_u$ without 
a preimage under $q_u \circ p_u$, the image $\nu(\gamma''(\pi''(x_u')))$ is 
undefined. Again by Lemma~\ref{lem:rule-instantiation-swap} we can assume 
w.l.o.g.~that it is the last quantification of $\iota_k$.

In this case $\nu' = \nu \circ q_u'' \circ p_u''$ satisfies the necessary 
conditions. Obviously $\nu'$ is injective and $\nu = \nu' \circ \mu_u''$ holds, 
so it remains to be shown that it is surjective. Assume there is an $y \in R$ 
without a preimage under $\nu'$. Since $\nu$ is injective and surjective, there 
is exactly one $\overline{y}_u \in \overline{R}_u$ with $\nu(\overline{y}_u) = 
y$. Because of commutativity, $\overline{x}_u$ cannot have a preimage under 
$q_u'' \circ p_u''$. Since the outer square is a pushout, there has to be an 
$y_u \in R_u$ with $\gamma''(\pi''(y_u)) = \overline{y}_u$. By assumption this 
$y_u$ has a preimage under $q_u \circ p_u$ (otherwise 
$\nu(\gamma''(\pi''(y_u)))$ would be undefined), which in turn has an image in 
$\overline{R}$. By commutativity $y$ must have a preimage under $\nu'$. Thus, 
$\nu'$ is surjective and a subgraph morphism. \qed
\end{proof}

\begin{lemma}\label{lem:minimal-pocs-finite}
Let $r : L \pto R$ be a partial morphism and let $m : R \ito G$ be total and 
injective. The set of pushout complements $G'$ of $m$ and $r$ where $m' : L 
\ito G'$ is injective has finitely many minimal elements and this minimal 
elements are computable.
\end{lemma}

\begin{proof}
In \cite{jk:minor-wqo} is was shown how the minimal pushout complements with 
respect to the minor ordering and conflict-free matching can be computed. The 
procedure can be easily adapted to subgraphs with injective matching. First the 
co-match $m$ can be assumed to injective, since a non-injective co-match 
implies a non-injective match $m'$ (see Lemma~\ref{lem:preserve-injection}). 
Furthermore a pushout complement with a non-injective $m'$ can be dropped since 
neither it nor any larger pushout complement has an injective match, as shown 
below.
\begin{center}
  \scalebox{1.0}{\begin{tikzpicture}[x=1.5cm, y=-1.5cm]

\node (H2) at (-1,1) {$H'$};
\node (L) at (0,0) {$L$};
\node (R) at (1,0) {$R$};
\node (H1) at (0,1) {$H$};
\node (G) at (1,1) {$G$};
\draw[mor-subgraph] (H2) -- node[midway, above] {$\mu$} (H1);
\draw[mor-tot] (L) -- node[midway, above left] {$m''$} (H2);
\draw[mor-parl] (L) -- node[midway, above] {$\gamma$} (R);
\draw[mor-tot] (L) -- node[midway, right] {$m'$} (H1);
\draw[mor-tot-inj] (R) -- node [midway, right] {$m$} (G);
\draw[mor-parl] (H1) -- node [midway, above] {$\gamma'$} (G);

\end{tikzpicture}%
}
\end{center}
Assume there are two pushout complements $H$, $H'$ with $H \subOrder H'$ and 
$m'$ is non-injective. There is a subgraph morphism $\mu : H' \subArrow H$ such 
that the diagram above commutes, especially $m' = \mu \circ m''$ holds. Since 
$m'$ is non-injective but $\mu$ is injective, $m''$ must be non-injective as 
well. Thus $H'$ can be dropped as well.

Since we obtain a subset of the minimal pushout complements of 
\cite{jk:minor-wqo}, the finiteness of this set is preserved. \qed
\end{proof}

\begin{lemma_app}{lem:proc-correctness1}
\lemmaProcCorrectnessOne
\end{lemma_app}

\begin{proof}
We have proven this statement in \cite{wsts-gts-framework} for rules without 
universal quantification and for conflict-free matches. The proof can be 
directly extended to this setting by using the fact that every injective match 
is 
automatically conflict-free. Since the number of instantiations is bounded by  
$\instBound[\rho]{S}$ for every rule $\rho$ and the number of subgraph 
morphisms $\mu : \overline{R} \subArrow R'$ is finite (up to isomorphism), the 
number of morphisms for which the pushout complement need to be computed is 
finite as well. Furthermore by Lemma~\ref{lem:minimal-pocs-finite} the set of 
minimal pushout complements is finite and computable, thus $\predBasis[k]{G}$ 
is a finite sets. \qed
\end{proof}

\begin{lemma_app}{lem:proc-correctness2}
\lemmaProcCorrectnessTwo
\end{lemma_app}

\begin{proof}
Let $G_0$ be an element of $\upclosed{\pred[\boundedPathk]{\upclosed{\{G\}}}}$. 
Then there is a minimal representative $G_1 \in \pred{\upclosed{\{G\}}}$ with 
$G_1 \subOrder G_0$ via some morphism $\nu : G_0 \subArrow G_1$ and an 
instantiation $(\pi : L \ito \overline{L}, \gamma : \overline{L} \pto 
\overline{R}$ of some rule $\rho$ rewriting $G_1$ with a injective match $m$ 
satisfying the application conditions of 
Definition~\ref{def:uqrule-application} to some element $G_2$ of 
$\upclosed{\{G\}}$. In \cite{wsts-gts-framework} it was shown that subgraph 
morphisms are pushout closed. Since $m$ is injective and therefore 
conflict-free, the left diagram below can be extended to the right diagram 
below, where the inner and outer squares are pushouts.
\begin{center}
  \scalebox{1.0}{\begin{tikzpicture}[x=1.5cm, y=-1.5cm]

\begin{scope}[shift={(0,0)}]
  \node (G1a) at (-1,1) {$G_0$};
  \node (La) at (0,0) {$\overline{L}$};
  \node (R1a) at (1,0) {$\overline{R}$};
  \node (G2a) at (0,1) {$G_1$};
  \node (G3a) at (1,1) {$G_2$};
  \node (Sa) at (2,2) {$G$};
  \draw[mor-subgraph] (G1a) -- node[midway, above] {$\nu$} (G2a);
  \draw[mor-parl] (La) -- node[midway, above] {$\gamma$} (R1a);
  \draw[mor-tot] (La) -- node[midway, right] {$m$} (G2a);
  \draw[mor-tot] (R1a) -- node [midway, right] {$m'$} (G3a);
  \draw[mor-parl] (G2a) -- node [midway, above] {$\gamma'$} (G3a);
  \draw[mor-subgraph] (G3a) -- node [midway, above right] {$\mu$} (Sa);
\end{scope}

\begin{scope}[shift={(4,0)}]
  \node (G1b) at (-1,1) {$G_0$};
  \node (Lb) at (0,0) {$\overline{L}$};
  \node (R1b) at (1,0) {$\overline{R}$};
  \node (G2b) at (0,1) {$G_1$};
  \node (G3b) at (1,1) {$G_2$};
  \node (Sb) at (2,2) {$G$};
  \node (R2b) at (2,0) {$\overline{R}'$};
  \node (G4b) at (0,2) {$G_3$};
  \draw[mor-subgraph] (G1b) -- node[midway, above] {$\nu$} (G2b);
  \draw[mor-parl] (Lb) -- node[midway, above] {$\gamma$} (R1b);
  \draw[mor-tot] (Lb) -- node[midway, right] {$m$} (G2b);
  \draw[mor-tot] (R1b) -- node [midway, right] {$m'$} (G3b);
  \draw[mor-parl] (G2b) -- node [midway, above] {$\gamma'$} (G3b);
  \draw[mor-subgraph] (G3b) -- node [midway, above right] {$\mu$} (Sb);
  \draw[mor-subgraph] (R1b) -- node [midway, above] {$\mu_R$} (R2b);
  \draw[mor-tot] (R2b) -- node [midway, right] {$n$} (Sb);
  \draw[mor-subgraph] (G2b) -- node [midway, right] {$\mu_G$} (G4b);
  \draw[mor-parl] (G4b) -- node [midway, above] {$s$} (Sb);
\end{scope}

\end{tikzpicture}%
}
\end{center}
Since $m$ and $\mu_G$ are injective, $\mu_G \circ m$ is injective as well and 
because of Lemma~\ref{lem:preserve-injection} we know that $n$ is also 
injective. Furthermore the pushout closure guarantees that $\mu_G \circ m$ is 
total. Since $m$ satisfied the application condition, every edge in $G_1$ 
incident to a universally quantified node has a preimage in $\overline{L}$ and 
therefore also an image in $G_3$. The surjectivity of $\mu_G$ ensures that the 
application condition is also satisfied by $\mu_G \circ m$. Note that since 
$G_1$ is an element of $\boundedPathk$ and $\boundedPathk$ is downward-closed, 
$G_3$ is also in $\boundedPathk$.

Since the outer square is a pushout, $G_3$ is a pushout complement object. 
Thus, a graph $G_4$ with $\mu_G' : G_4 \subOrder G_3$ will be obtained by the 
procedure \predBasis[k]{} in Step~\ref{backstep:calculate-pocs} using the 
instantiation $\mu_R \circ \gamma$. By the same argument as above $\mu_G' \circ 
\mu_G \circ m$ satisfies the application condition and is an element of 
$\boundedPathk$, thus $G_4$ will not be dropped by the procedure. Summarized, 
this means that \predBasis[k]{} computes a graph $G_4$ for every graph $G_0$ 
such that $G_4 \subOrder G_3 \subOrder G_1 \subOrder G_0$, i.e.~every $G_0$ is 
represented by an element of \predBasis[k]{G}. \qed
\end{proof}

\begin{lemma_app}{lem:pre-to-post-condition}
\lemmaPreToPostCondition
\end{lemma_app}

\begin{proof}
Assume there is a $x \in \qNodes(u)$ where $x' = m(\gamma(\pi(x)))$ is defined 
and there is an edge $e$ attached to $x'$ without preimage in $\overline{R}$. 
Furthermore, assume $H$ with $m' : \overline{L} \to H$ and $\gamma' : H \pto G$ 
is a pushout complement of $\gamma$, $m$. Since the diagram is a pushout, there 
is an $e' \in H$ with $\gamma'(e') = e$, otherwise the mediating morphism does 
not exist or is not unique. By commutativity of the diagram, $e'$ is attached 
to $m'(\pi(x))$ and there cannot be an $e'' \in \overline{L}$ with $m'(e'') = 
e'$. Since $x \in \qNodes(u)$, this violates the condition of 
Definition~\ref{def:uqrule-application}. \qed
\end{proof}

\begin{lemma_app}{lem:better-bound}
\lemmaBetterBound
\end{lemma_app}

\begin{proof}
This can be shown by using the proof ideas of 
Proposition~\ref{prop:inst-bound-exists}. If $\nu$ does not delete all elements 
of $\overline{R}_u$ which where created in the instantiation step, 
$\overline{R}_u$ contains at least one node more than $R$. If the created 
element not deleted by $\nu$ is an edge, by conditions of this lemma, it is 
incident to a created node not deleted by $\nu$. Thus, every non-negligible 
instantiation step increases the number of nodes of the right side by at least 
one. No matchings can exist if the number of instantiation steps is larger than 
the number of nodes in $G$. \qed
\end{proof}
}

\end{document}